\documentclass[10pt,journal,compsoc]{IEEEtran}
\usepackage{amsmath,amssymb,amsfonts}
\usepackage{amsmath, amsthm, amssymb}
\usepackage{graphicx}
\usepackage{textcomp}
\usepackage{xcolor}
\usepackage{subcaption}
\usepackage{float}
\usepackage{url}
\usepackage[inline]{enumitem}
\usepackage[symbol]{footmisc}
\usepackage[ruled,vlined]{algorithm2e}
\usepackage{mathtools}
\usepackage{cleveref}
\usepackage{textcomp}
\newtheorem{prop}{Proposition}
\newtheorem{remark}{Remark}

\ifCLASSOPTIONcompsoc
  \usepackage[nocompress]{cite}
\else
  \usepackage{cite}
\fi

\ifCLASSINFOpdf
\else
\fi

\begin{document}
\title{DynaMo: Dynamic Community Detection by Incrementally Maximizing Modularity}
\author{Di~Zhuang,
        J.~Morris~Chang,~\IEEEmembership{Senior~Member,~IEEE}
        and~Mingchen~Li
\IEEEcompsocitemizethanks{\IEEEcompsocthanksitem The authors are with the Department of Electrical Engineering, University of South Florida, Tampa, FL 33620.\protect\\
E-mail: \{zhuangdi1990, morrisjchang, mingchenli1992\}@gmail.com.}
}

\IEEEtitleabstractindextext{%
\begin{abstract}
Community detection is of great importance for online social network analysis. The volume, variety and velocity of data generated by today's online social networks are advancing the way researchers analyze those networks. For instance, real-world networks, such as Facebook, LinkedIn and Twitter, are inherently growing rapidly and expanding aggressively over time. However, most of the studies so far have been focusing on detecting communities on the static networks. It is computationally expensive to directly employ a well-studied static algorithm repeatedly on the network snapshots of the dynamic networks. We propose DynaMo, a novel modularity-based dynamic community detection algorithm, aiming to detect communities of dynamic networks as effective as repeatedly applying static algorithms but in a more efficient way. DynaMo is an adaptive and incremental algorithm, which is designed for incrementally maximizing the modularity gain while updating the community structure of dynamic networks. In the experimental evaluation, a comprehensive comparison has been made among DynaMo, Louvain (static) and 5 other dynamic algorithms. Extensive experiments have been conducted on 6 real-world networks and 10,000 synthetic networks. Our results show that DynaMo outperforms all the other 5 dynamic algorithms in terms of the effectiveness, and is 2 to 5 times (by average) faster than Louvain algorithm.
\end{abstract}

}

\maketitle

\IEEEdisplaynontitleabstractindextext

%
\IEEEpeerreviewmaketitle

\IEEEraisesectionheading{\section{Introduction}\label{sec:introduction}}
\IEEEPARstart{W}{ith} the advance of online social network analysis, more and more real-world systems, such as 
social networks \cite{de2018exploratory}, collaboration relationships \cite{zhou2017collaboration}, recommendation systems \cite{ying2018graph} and intrusion detection system \cite{zhuang2017peerhunter, zhuang2019enhanced}, are represented and analyzed as networks, where the vertices represent certain objects and the edges represent the relationships or connections between the objects. Most social networks have been shown to present certain community structures \cite{he2018hidden}, 
where vertices are densely connected within communities and sparsely connected between communities. Community detection is one of the most important and fundamental problem in the field of graph mining, network science and social network analysis.

Detecting community structure is of great challenge, and most of the recent studies are proposed to detect communities in the static networks, such as spectral clustering \cite{zhang2018understanding}, 
label propagation \cite{meng2018semi}, 
modularity optimization \cite{newman2006finding}, 
and k-clique communities \cite{hao2018detecting}. 
However, real-world networks, especially most of the online social networks, are not static. Most popular online social networks (e.g., Facebook, LinkedIn and Twitter) are de facto growing rapidly and expanding aggressively in terms of either the size or the complexity over time. For instance, in Facebook network, the updating of its community structure could be simply caused by new users joining in, old users leaving, or certain users connecting (i.e., friend) or disconnecting (i.e., unfriend) with the other users. Facebook announced that it had 1.52 billion daily active users in the fourth quarter of 2018 \cite{FacebookDAUs}, which shows a 9\% increase over the same period of the previous year, and 4 million likes generated every minute as of January 2019 \cite{FacebookFacts}. Hence, it is rather important and impending to enable community detection in such dynamic networks.

Designing an effective and efficient algorithm to detect communities in dynamic networks is highly difficult.
First, an efficient algorithm should update the communities adaptively and incrementally depending on the changes of the dynamic networks, and avoid redundant and repetitive computations.
Second, it is hard to design a dynamic algorithm that performs as effective as the static algorithms by only observing the historical community structures and the incremental changes of the dynamic networks.
Third, it is still quite open about how to categorize the incremental changes of dynamic networks, and how to assess the influence of different types of the incremental changes on the community structure updates, which is rather important to design an effective and efficient dynamic algorithm.

A few algorithms have been proposed to detect communities in dynamic networks 
\cite{dhouioui2014tracking, ilhan2015predicting, nguyen2011adaptive, shang2014real, xin2016adaptive, shang2016targeted, agarwal2018dyperm, chong2013incremental, rossetti2017tiles}.
An intuitive way to detect communities in dynamic networks is to slice the network into small snapshots based on the timestamps, and directly employ well-studied static algorithms repeatedly on each network snapshot. However, these algorithms 
\cite{dhouioui2014tracking, ilhan2015predicting} usually are computational expensive, since they compute the current community structures completely independent from the historical information (i.e., the previous community structures), especially when the dynamic network changes rapidly and the time interval between two consecutive network snapshots are extremely small. Another way to update the communities is using not only the current network changes but also the previous community structures.
These algorithms \cite{nguyen2011adaptive, shang2014real, xin2016adaptive, shang2016targeted, agarwal2018dyperm, chong2013incremental, rossetti2017tiles} adaptively and incrementally detect communities in dynamic networks, without re-executing any static algorithms on each entire network snapshot. Those algorithms are usually more efficient than repeatedly applying static algorithms on network snapshots. However, most of those algorithms are still not practical enough to be directly used to analyze the real-world networks. For instance, some algorithms \cite{nguyen2011adaptive, shang2016targeted} only considers vertices/edges additions, while vertices/edges deletions happen quite often in online social networks (e.g., ``unfriend'' in Facebook). Some algorithms \cite{nguyen2011adaptive, shang2014real, shang2016targeted, chong2013incremental} only consider unweighted networks, which are not applicable for weighted networks. Furthermore, some algorithms \cite{lin2009analyzing, lin2011community, rossetti2017tiles} need certain prior information about the community structures (e.g., the number of communities, the ratio of vertices in overlapped communities) or need certain predefined parameters which are not available in practice.

We present DynaMo, a novel modularity-based dynamic community detection algorithm, aiming to detect non-overlapped communities of dynamic networks. DynaMo is an adaptive and incremental algorithm designed for maximizing the modularity gain while updating the community structure of dynamic networks. To update the community structures efficiently, we model the dynamic network as a sequence of incremental network changes. We propose 6 types of incremental network changes: (a) intra-community edge addition/weight increase, (b) cross-community edge addition/weight increase, (c) intra-community edge deletion/weight decrease, (d) cross-community edge deletion/weight decrease, (e) vertex addition, and (f) vertex deletion. For each incremental network change, we design an operation to maximize the modularity.

In the experimental evaluation, a comprehensive comparison has been made among DynaMo, Louvain (static) \cite{blondel2008fast} and 5 dynamic algorithms (i.e., QCA \cite{nguyen2011adaptive}, Batch \cite{chong2013incremental}, GreMod \cite{shang2014real}, LBTR-LR \cite{shang2016targeted} and LBTR-SVM \cite{shang2016targeted}). Extensive experiments have been conducted on 6 large-scale real-world networks and 10,000 synthetic networks. Our results show that DynaMo consistently outperforms all the other 5 dynamic algorithms in terms of the effectiveness, and is 2 to 5 times (by average) faster than Louvain algorithm. To summarize, our work has the following contributions:

$\bullet$ We present a novel, effective and efficient modularity-based dynamic community detection algorithm, DynaMo, capable of detecting non-overlapped communities in real-world dynamic networks.

$\bullet$ We present the theoretical analysis to show why/how DynaMo could maximize the modularity, while avoiding certain redundant and repetitive computations.

$\bullet$ A comprehensive comparison among our algorithm and the state-of-the-art algorithms has been conducted (Section~\ref{sec:ExperimentalEvaluation}). For the sake of reproducibility and convenience of future studies about dynamic community detection, we have released our prototype implementation of DynaMo, the experiment datasets and a collection of the implementations of the other state-of-the-art algorithms. \footnote[1]{\url{https://github.com/nogrady/dynamo}}

The rest of this paper is organized as follows:
Section~\ref{sec:RelatedWork} presents the related work.
Section~\ref{sec:Preliminaries} presents the notations, the concept of dynamic networks and the definition of modularity, and introduces a baseline static community detection algorithm (i.e., Louvain algorithm).
Section~\ref{sec:DynaMo} describes our algorithm design and theoretical analysis.
Section~\ref{sec:ExperimentalEvaluation} presents the experimental evaluation.
Section~\ref{sec:Conclusion} concludes.

\section{Related Work} \label{sec:RelatedWork}
To date, a few dynamic community detection approaches were proposed
\cite{greene2010tracking, dhouioui2014tracking, ilhan2015predicting, lin2009analyzing, lin2011community, tang2012identifying, guo2014evolutionary, zakrzewska2015dynamic, rossetti2017tiles, ma2017evolutionary, nguyen2011adaptive, chong2013incremental, shang2014real, shang2016targeted, cordeiro2016dynamic}.
Rossetti et al. \cite{rossetti2018community}, a comprehensive survey on dynamic community detection, divide most of the algorithms into three categories (i.e., instant-optimal, temporal trade-off and cross-time) in terms of their ability to solve the community instability and temporal smoothing issue. However, some approaches in the literature are not belonging to any of these categories. For instance, some approaches, such as our proposed algorithm, do not consider the community instability and temporal smoothing issue, but still aim to detect communities in dynamic networks effectively and efficiently. The concept of ``cross-time'' approaches is also beyond the scope of this paper. After incorporating the taxonomy of \cite{rossetti2018community}, we consider three categories of related approaches: instant-optimal, temporal trade-off and incremental approaches (to replace the ``cross-time'' in \cite{rossetti2018community}).

The instant-optimal approaches \cite{greene2010tracking, dhouioui2014tracking, ilhan2015predicting} have two steps: (i) static algorithms are applied on each network snapshot independently to detect static communities, (ii) communities detected on each network snapshot are matched with communities detected on the previous one. Greene et al. \cite{greene2010tracking} proposed a general model for tracking communities in dynamic networks via solving a classic cluster matching problem on the communities independently detected on consecutive network snapshots. Such approaches take advantage of existing static algorithms. However, repeatedly applying static algorithms on all network snapshots of the dynamic networks 
is computationally expensive.


The temporal trade-off approaches \cite{lin2009analyzing, lin2011community, tang2012identifying, guo2014evolutionary, zakrzewska2015dynamic, rossetti2017tiles, ma2017evolutionary} incorporate the community detection and tracking via considering the community structures of the current and historical network snapshots at the same time. Those approaches aim to maintain the evolution of the community structures of the dynamic networks, where the community structure (e.g., the number of communities, the size of communities) of the current network snapshot should be similar to that of the previous one. Tang et al. \cite{tang2012identifying} propose a temporally regularized clustering algorithm to identify evolving groups in dynamic networks, where they use a metric that attempts to optimize two objectives: the quality of the current community structure and the similarity between the current and the previous community structures. However, most of those approaches, such as \cite{lin2009analyzing, lin2011community, ma2017evolutionary}, require determining the number of communities to be detected/tracked in advance, which is rather impractical for the real-world dynamic networks where the number of communities changes over time. 

The incremental approaches \cite{nguyen2011adaptive, chong2013incremental, shang2014real, shang2016targeted, cordeiro2016dynamic} adaptively update the community structures fully based on the network changes happened during the current snapshot and the community structure of the previous snapshot. 
For instance, GreMod \cite{shang2014real} is a rule-based incremental algorithm that performs the predetermined operations on different types of the edge addition changes of the dynamic network. QCA \cite{nguyen2011adaptive} is another rule-based adaptive algorithm that updates the community structures according to the predefined rules of different types of the incremental changes (i.e., vertices/edges addition/deletion) on the dynamic network. QCA is also one of the most efficient dynamic community detection algorithms in the literature. However, since the rule-based algorithms, such as GreMod \cite{shang2014real} and QCA \cite{nguyen2011adaptive}, considers each network change as an independent event,
they are less efficient when abundant (i.e., a batch of) network changes appear in the same network snapshot.
Chong et al. \cite{chong2013incremental} propose a batch-based incremental modularity optimization algorithm that updates the community structures by initializing all of the new and changed vertices of the current network snapshot (i.e., the batch) as singleton communities and using Louvain algorithm to further update the community structures. However, since their initialization approach, that generates the intermediate community structure of a batch of network changes, is rather coarse, it is less efficient to apply Louvain algorithm on those intermediate community structures. LBTR \cite{shang2016targeted} is a learning-based framework that uses machine learning classifiers and historical community structure information to predict certain vertices' new community assignments after each round of network changes. In those learning-based algorithms, once the models are being trained, the testing phase could be very efficient. However, since the supervised nature of the learning-based algorithms, it would be extremely hard to generalize the trained models. For instance, the models trained on one type of dynamic networks (e.g., social network) might be less effective to another type of dynamic networks (e.g., collaboration network). Furthermore, even for the same dynamic network, the network patterns change over time. Thus, the models have to be updated periodically, which would be rather illogical, since the network usually changes rapidly and updating models is also time consuming.

Our proposed approach, DynaMo, is an adaptive and incremental algorithm. Compared with rule-based algorithms \cite{shang2014real, nguyen2011adaptive}, our approach is capable of processing a set of network changes as a batch, and redesigns the ``rules'' by considering more extreme cases (Section~\ref{sec:DynaMoAlgorithm}). Compared with batch-based algorithms \cite{chong2013incremental}, our approach has a more fine-grained initialization phase (Section~\ref{sec:DynaMoAlgorithm}), which could reduce the computation time dramatically. Compared with learning-based algorithms \cite{shang2016targeted}, our approach is more generalized to real-world networks.
In Section~\ref{sec:ExperimentalEvaluation}, we compare DynaMo with Louvain algorithm and 5 other dynamic algorithms on 6 real-world networks and 10,000 synthetic networks, showing that DynaMo consistently outperforms all the other 5 dynamic algorithms in terms of effectiveness, and much more efficient than Louvain algorithm.

\section{Preliminaries} \label{sec:Preliminaries}
In this section, we introduce 1) the notations; 2) the dynamic network model; 3) modularity, to quantify the quality of a community structure; and 4) Louvain algorithm, a modularity-based static community detection approach.

\subsection{Notations}
Let $G=(V, E)$ be an undirected weighted network, where $V$ is a set of vertices ($n=|V|$), $E$ is a set of undirected weighted edges ($m=|E|$), and there could be more than one edge between a pair of vertices. Let $C$ denote a set of disjoint communities associated with $G$, $A_{ij}$ denote the sum of the weights of all the edges between vertices $i$ and $j$, $k_{i}$ denote the sum of the weights of all the edges linked to vertex $i$, and $c_{i}$ denote the assigned community of vertex $i$.

\subsection{Dynamic Network}
Let $G^{(t)}$ denote the snapshot of a network at time $t$, and $\triangle G^{(t)}=(\triangle V^{(t)}, \triangle E^{(t)})$ denote the incremental change from $G^{(t)}$ to $G^{(t+1)}$ (i.e., $G^{(t+1)}=G^{(t)} \cup \triangle G^{(t)}$), where $\triangle V^{(t)}$ and $\triangle E^{(t)}$ are the sets of vertices and edges 
being changed during time period $(t, t+1]$. A dynamic network $G$ is a sequence of its network snapshots changing over time: $G=\{G^{(0)},$ $G^{(1)},$ $\dots,$ $G^{(t)}\}$.

\subsection{Modularity} \label{sec:Preliminaries_Modularity}
Modularity~\cite{newman2006modularity} is a widely used criteria to evaluate the quality of given network community structure. Community structures with high modularity have denser connections among vertices in the same communities but sparser connections among vertices from different communities. Given network $G=(V, E)$, its modularity is defined as follows:
\begin{equation}
    \begin{split} \label{eq:modularity}
        Q=\frac{1}{2m}\sum_{i, j \in V}[A_{ij}-\frac{k_{i}k_{j}}{2m}]\delta_{ij}
        =\frac{1}{2m}{\sum_{c \in C}^{c}}(\alpha_{c}-\frac{\beta_{c}^{2}}{2m})
    \end{split}
\end{equation}
where $\alpha_{c}=\sum_{i, j \in c} A_{ij}$, $\beta_{c}=\sum_{i \in c} k_{i}$ and $\delta_{ij}$ equals to 1, if $i$, $j$ belong to the same community, otherwise equals to 0.

\subsection{Louvain Method for Community Detection} \label{sec:Preliminaries_Louvain}
Since the modularity optimization problem is known to be NP-hard, various heuristic approaches are proposed \cite{clauset2004finding, duch2005community, newman2006finding}. Most of the algorithms have been superseded by Louvain algorithm~\cite{blondel2008fast}, which attempts to maximize the modularity using a greedy optimization approach composed of three steps:
(i) {\bf Initialization}, where each vertex forms a singleton community. (ii) {\bf Local Modularity Optimization}, where each vertex moves from its own community to its neighbor's community to maximize the local modularity gain. If there is no positive modularity gain, keep the vertex in its original community. Repeat this step over all vertices multiple times until the modularity gain is negligible. (iii) {\bf Network Compression}, where vertices belonging to the same community are aggregated as super vertices and a new network is built with the super vertices.


Louvain algorithm repeats the last two steps, until the modularity improvements is negligible. Although the actual computational complexity of Louvain depends on the input network, it has an average-case time complexity of $O(m)$ with most of the computational effort spending on the optimization of the first level network (i.e., before creating the super vertices).

\section{DynaMo: Dynamic Community Detection by Incrementally Maximizing Modularity} \label{sec:DynaMo}

\begin{figure}[tb]
\centering
\includegraphics[width=255pt]{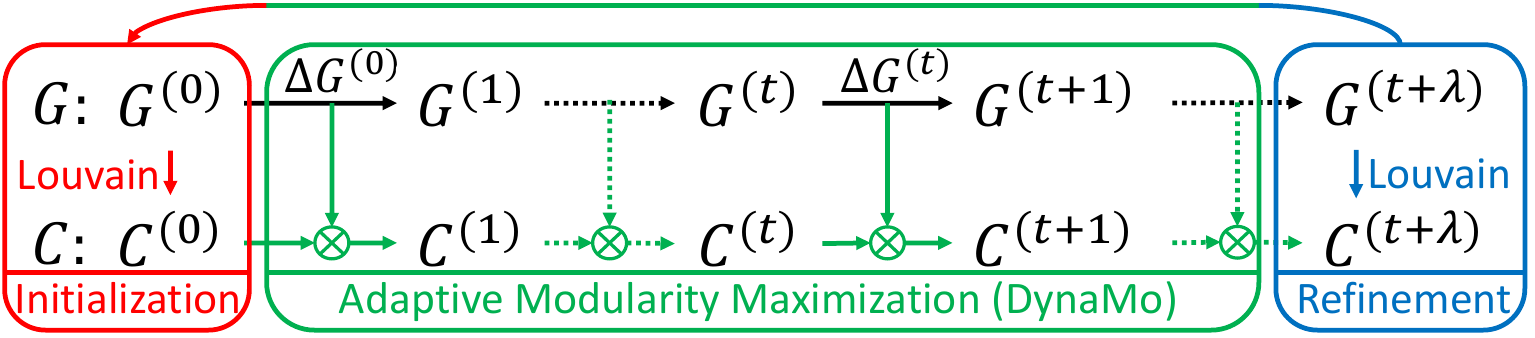}
\caption{The overview of DynaMo.}
\label{fig:Evolve}
\end{figure}

\subsection{Problem Statement} \label{sec:DynaMo_ProblemStatement}
Given a dynamic network $G=\{G^{(0)},$ $G^{(1)},$ $\dots,$ $G^{(t)}\}$, where $G^{(0)}$ is the initial network snapshot, let $C=\{C^{(0)},$ $C^{(1)},$ $\dots,$ $C^{(t)}\}$ denote the list of community structures of the corresponding network snapshots. As illustrated in Fig.~\ref{fig:Evolve}, we aim to design an adaptive and incremental algorithm to detect $C^{(t+1)}$, given $G^{(t)}$, $C^{(t)}$ and $\triangle G^{(t)}$.

\subsection{Methodology Overview} \label{sec:DynaMo_MethodologyOverview}
As shown in Fig.~\ref{fig:Evolve}, our approach has three components: 

\begin{itemize}
    \item \textit{Initialization}: Use well-studied static algorithms (i.e., Louvain~\cite{blondel2008fast}) to compute $C^{(0)}$, which generates a comparatively accurate community structure of $G^{(0)}$.
    \item \textit{Adaptive Modularity Maximization (DynaMo)}:  Given $G^{(t)}$, $C^{(t)}$ and $\triangle G^{(t)}$, update the community structure of $G^{(t+1)}$ from $C^{(t)}$ to $C^{(t+1)}$ while maximizing the modularity gain, using predesigned strategies that fully depend on 
        $\triangle G^{(t)}$ and 
        $C^{(t)}$. This is the core component of our framework that relies on fine-grained and theoretical-verified strategies (Section~\ref{sec:DynaMoAlgorithm}) to maximize the modularity gain while maintaining the efficiency. 
	\item \textit{Refinement}: Once the obtained modularity of $C^{(t+\lambda)}$ is less than a predefined threshold, use $G^{(t+\lambda)}$ as the new initial network snapshot to restart our algorithm from the initialization step. This component prevents our frame from being trapped in the suboptimal solutions.
\end{itemize}

\subsection{The DynaMo Algorithm}
\label{sec:DynaMoAlgorithm}
DynaMo is an adaptive and incremental algorithm aiming to maximize the community structure modularity gain based on the incremental changes of a dynamic network. We propose a two-step approach: (i) initialize an intermediate community structure, depending on the incremental network changes and the previous network community structure, and (ii) repeat the last two steps of Louvain algorithm (Section~\ref{sec:Preliminaries_Louvain}) on the intermediate community structure until the modularity gain is negligible.

Our algorithm benefits community detection in dynamic networks in three folds.
First, in the initialization step, we categorize the incremental changes into 6 types. For each type of the incremental change, we design a strategy to initialize its corresponding intermediate community structure. Most of the strategies are theoretically verified to incrementally maximize the modularity, while avoiding redundant and repetitive computations.
Second, compared with the original initialization step of Louvain algorithm, our initialization step takes advantage of the historical information, thus reduces most of the unnecessary computations happened at Louvain's first level network optimization, where Louvain spends most of its computational effort (Section~\ref{sec:Preliminaries_Louvain}). Hence, DynaMo would be much more efficient than Louvain algorithm while detecting communities in dynamic networks.
Third, in the initialization, our algorithm could process a set of incremental changes as a batch, which makes the computational complexity of our algorithm less sensitive to the amount of incremental changes and the frequency of network changes. So, DynaMo can detect communities while the network changing rapidly.

In this section, 6 different types of the incremental changes have been defined, where the initialization strategy of each type is also designed accordingly. Eight propositions are proposed and proved to provide the theoretical guarantees of our strategies towards maximizing the modularity.

\subsubsection{Edge Addition/Weight Increase (EA/WI)}
In this scenario, an edge $(i, j, w_{ij})$ between two existing vertices $i$ and $j$ has been changed to $(i, j, w_{ij}+\triangle w)$, where $w_{ij} \geq 0$ and $\triangle w > 0$. Edge addition is a special case of edge weight increase, where $w_{ij}=0$. Depending on the edge property, we define two sub-scenarios:

{\bf Intra-Community EA/WI (ICEA/WI):} Vertices $i$ and $j$ belong to the same community (i.e., $c_{i}=c_{j}$). According to Proposition~\ref{prop:1}, ICEA/WI will never split $i$ and $j$ into different communities. And according to Remark~\ref{remark:1}, sometimes ICEA/WI will split $c_{i}$ into multiple communities, while keeping $i$ and $j$ in the same community. Proposition~\ref{prop:2} also provides us a convenient tool to decide when $c_{i}$ should be bi-split into two smaller communities (i.e., $c_{p}$ and $c_{q}$). However, this approach requires checking all the bi-split combinations of $c_{i}$, which is time consuming, especially when $c_{i}$ is huge. In this case, we propose to initialize $i$ and $j$ as a two-vertices community, and all the other vertices in $c_{i}$ as singleton communities.

\begin{figure*}[h]
\captionsetup{font=footnotesize}
        \centering
        \begin{subfigure}[b]{0.5\textwidth}
                \includegraphics[width=\textwidth]{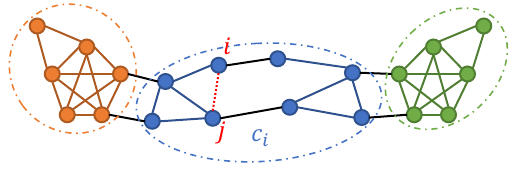}
                \caption{}
                \label{fig:ICEA_SP_CP_1}
        \end{subfigure}%
        ~ 
        \begin{subfigure}[b]{0.5\textwidth}
                \includegraphics[width=\textwidth]{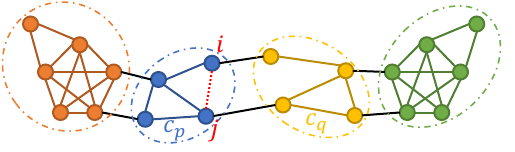}
                \caption{}
                \label{fig:ICEA_SP_CP_2}
        \end{subfigure}
        \caption{Two possible behaviors of the community structure after adding an intra-community edge: (a) unchanged, (b) splitting to smaller communities.}
        \label{fig:ICEA_SP_CP}
\end{figure*}

\begin{prop} \label{prop:1}
Adding an edge or increasing the edge weight between vertices $i$ and $j$, that belong to the same community ($c_{i}=c_{j}$), will not split $i$ and $j$ into different communities.
\end{prop}

See Appendix \ref{Appendix_A_1} for the proof.

\begin{remark} \label{remark:1}
Although our Proposition 1 shows that ICEA/WI between $i$ and $j$, where $c_{i}=c_{j}$, will not split them into different communities, sometimes splitting $c_{i}$ into smaller communities in other ways (i.e., keeping $i$ and $j$ in the same community after the splitting) might maximize the modularity. For instance, as shown in Fig.~\ref{fig:ICEA_SP_CP}, assume all the edge weights are 1.0, and the red dash line between $i$ and $j$ is a newly added intra-community edge. Before adding the new edge, the modularity of community structure in Fig.~\ref{fig:ICEA_SP_CP_1} (i.e., $0.561$) is higher than that in Fig.~\ref{fig:ICEA_SP_CP_2} (i.e., $0.558$). However, after adding the new edge, the modularity of community structure in Fig.~\ref{fig:ICEA_SP_CP_1} (unchanged, i.e., $0.564$) becomes lower than that in Fig.~\ref{fig:ICEA_SP_CP_2} (split, i.e., $0.568$). In this case, although an intra-community edge has been added, splitting $c_{i}$ into $c_{p}$ and $c_{q}$ provides higher modularity. Our algorithm carefully considers these ``counterintuitive'' cases, which is different from QCA \cite{nguyen2011adaptive, nguyen2014dynamic}, thus, leading our algorithm to be more effective (Section~\ref{sec:ExperimentalEvaluation_Effectiveness}).
\end{remark}

\begin{prop} \label{prop:2}
(ICEA/WI Community Bi-split) After ICEA/WI between vertices $i$ and $j$, where $c_{i}=c_{j}$, if a bi-split of $c_{i}$ (i.e., $c_{p} \subseteq c_{i}$ and $c_{q} = c_{i} \backslash c_{p}$) does not exist such that $\triangle w > \frac{m\alpha_{1}-\beta_{c_{p}}\beta_{c_{q}}}{2\beta_{c_{q}}-\alpha_{1}}$, where $\alpha_{1}=\alpha_{c_{i}} - \alpha_{c_{p}} - \alpha_{c_{q}}$, any other bi-split of $c_{i}$ will not improve the modularity gain comparing with keeping the community structure unchanged.
\end{prop}

See Appendix \ref{Appendix_A_2} for the proof.


{\bf Cross-Community EA/WI (CCEA/WI):} Vertices $i$ and $j$ are from two different communities (i.e., $c_{i} \neq c_{j}$). CCEA/WI between $i$ and $j$ leads to three possible operations: (a) keeping the community structure unchanged; (b) merging $c_{i}$ and $c_{j}$ into one community; and (c) splitting $c_{k} = c_{i} \cup c_{j}$ into other smaller communities. 
According to Proposition~\ref{prop:4}, if $\triangle w$ is large enough, merging $c_{i}$ and $c_{j}$ into one community (e.g., $c_{k}$) provides higher modularity gain than keeping the community structure unchanged. However, if $\triangle w$ is too large (as shown in Proposition~\ref{prop:5}), CCEA/WI is equivalent to a two-step process: (a) CCEA/WI between $i$ and $j$ ($c_{i} \neq c_{j}$), that results in merging $c_{i}$ and $c_{j}$ into one community $c_{k}$ (Proposition~\ref{prop:4}); (b) ICEA/WI between $i$ and $j$ ($c_{i} = c_{j} = c_{k}$). Proposition~\ref{prop:5} provides a bi-split condition. 
However, Proposition~\ref{prop:5} also requires checking all bi-split combinations of $c_{k}$. Hence, to deal with CCEA/WI, we propose: (a) if $\triangle w \leq \frac{1}{2} \big( \alpha_{2} + \beta_{2} - 2m + \sqrt{(2m - \alpha_{2} - \beta_{2})^{2} + 4(m\alpha_{2}+\beta_{c_{i}}\beta_{c_{j}}}) \big)$, where $\alpha_{2}=\alpha_{c_{i}}+\alpha_{c_{j}}-\alpha_{c_{k}}$ and $\beta_{2}=\beta_{c_{i}} + \beta_{c_{j}}$, we keep the community structure unchanged; (b) otherwise, we employ the
same initialization approach proposed to deal with ICEA/WI on $c_{k} = c_{i} \cup c_{j}$, where we consider ICEA/WI has happened between vertices $i$ and $j$, where $c_{i}=c_{j}=c_{k}$.

\begin{prop} \label{prop:4} (CCEA/WI Community Merge)
After CCEA/WI between $i$ and $j$, where $c_{i} \neq c_{j}$, if and only if $\triangle w > \frac{1}{2} \big( \alpha_{2} + \beta_{2} - 2m + \sqrt{(2m - \alpha_{2} - \beta_{2})^{2} + 4(m\alpha_{2}+\beta_{c_{i}}\beta_{c_{j}}}) \big)$, where $\alpha_{2}=\alpha_{c_{i}}+\alpha_{c_{j}}-\alpha_{c_{k}}$ and $\beta_{2}=\beta_{c_{i}} + \beta_{c_{j}}$, merging $c_{i}$ and  $c_{j}$ into $c_{k}$ (i.e., $c_{k} = c_{i} \cup c_{j}$) has higher modularity gain than keeping the community structure unchanged.
\end{prop}

See Appendix \ref{Appendix_A_3} for the proof.


\begin{prop} \label{prop:5} (CCEA/WI Community Bi-split)
After CCEA/WI between $i$ and $j$, where $c_{i} \neq c_{j}$, $c_{k} = c_{i} \cup c_{j}$, and $\{c_{p}$, $c_{q}\}$ is another bi-split of $c_{k}$ (i.e., $c_{p} \subseteq c_{k}$ and $c_{q} = c_{k} \backslash c_{p}$), if and only if $\triangle w > \frac{1}{2} \big( \alpha_{2} + \beta_{2} - 2m + \sqrt{(2m - \alpha_{2} - \beta_{2})^{2} + 4(m\alpha_{2}+\beta_{c_{i}}\beta_{c_{j}}}) \big) + \frac{m \alpha_{1}-\beta_{c_{p}}\beta_{c_{q}}}{2\beta_{c_{q}}-\alpha_{1}}$,
where $\alpha_{1}=\alpha_{c_{i}} - \alpha_{c_{p}} - \alpha_{c_{q}}$, $\alpha_{2}=\alpha_{c_{i}}+\alpha_{c_{j}}-\alpha_{c_{k}}$ and $\beta_{2}=\beta_{c_{i}} + \beta_{c_{j}}$, splitting $c_{k}$ into $c_{p}$ and $c_{q}$ has higher modularity gain than either keeping the community structure unchanged or merging $c_{i}$ and  $c_{j}$ into $c_{k}$.
\end{prop}

The proof could be easily derived from Proposition~\ref{prop:2} and Proposition~\ref{prop:4}.

\subsubsection{Edge Deletion/Weight Decrease (ED/WD)} \label{sec:DynaMoAlgorithm_EDWD}
In this scenario, an edge $(i, j, w_{ij})$ between two existing vertices $i$ and $j$ has been changed to $(i, j, w_{ij}-\triangle w)$, where $w_{ij} \geq \triangle w > 0$. Edge deletion is a special case of edge weight decrease, where $w_{ij} = \triangle w$. 
Depending on the edge property, we define two sub-scenarios:

{\bf Intra-Community ED/WD (ICED/WD):} Vertices $i$ and $j$ belong to the same community (i.e., $c_{i}=c_{j}$). According to Proposition~\ref{prop:3}, if $i$ or $j$ has one degree, decreasing the edge weight between $i$ and $j$ will keep the community structure unchanged. Also, intuitively, if $i$ or $j$ has one degree, deleting the edge between $i$ and $j$ will result in the same community structure plus one or two singleton communities (i.e., the vertex of one degree becomes singleton community). Except for the case above (i.e., $i$ or $j$ has one degree), ICED/WD between $i$ and $j$ leads to three other possible operations: (a) keeping the community structure unchanged, if $c_{i}$ is still densely connected; (b) splitting $c_{i}$ into multiple smaller communities, if $c_{i}$ becomes sparsely connected; and (c) merging $c_{i}$ with some of its neighbor communities (i.e., the opposite situation of Remark~\ref{remark:1}). Since the analytical approach is complex and time consuming, we propose to initiate all vertices within the communities, that adjacent to $i$ or $j$ (including $c_{i}$), as singleton communities.

\begin{prop} \label{prop:3}
For any pair of vertices $i$, $j$ that belong to the same community (i.e., $c_{i}=c_{j}$), if $i$ or $j$ has only one neighbor vertex ($j$ or $i$), decreasing the edge weight between $i$ and $j$, does not split $i$ and $j$ into different communities.
\end{prop}

See Appendix \ref{Appendix_A_5} for the proof.

{\bf Cross-Community ED/WD (CCED/WD):} Vertices $i$ and $j$ are from two different communities (i.e., $c_{i} \neq c_{j}$). By Proposition~\ref{prop:6}, CCED/WD strengthens the community structure, thus, keeping the community structure unchanged.

\begin{prop} \label{prop:6}
If vertices $i$ and $j$ are from different communities ($c_{i} \neq c_{j}$), deleting an edge or decreasing the edge weight between $i$ and $j$, will increase the modularity gain coming from $c_{i}$ and $c_{j}$.
\end{prop}

See Appendix \ref{Appendix_A_6} for the proof.


\subsubsection{Vertex Addition (VA)}
In this scenario, a new vertex $i$ and its associated edges are added. On one hand, if $i$ has no associated edge, we make it as a singleton community and keep the rest community structure unchanged. On the other hand, if $i$ has one or more associated edges, some interesting cases would happen. For instance, if all of $i$'s associated edges are connected to the same community, i.e., $c_{j}$, by Proposition~\ref{prop:7}, we should merge $i$ into $c_{j}$ and treat all of $i$'s associated edges as ICEA/WI. A more complicated case occurs when $i$'s associated edges are connected to different communities. In this case, by Proposition~\ref{prop:8}, we could merge $i$ into community $c_{j}$ that has the highest $\triangle w_{ij}$. However, other than simply determining which community $i$ should merge into, we should also consider which set of vertices could together with $i$ to form a new community, or which community could be split into smaller communities, to further maximize the modularity. To cope with all the cases, where $i$ has one or more associated edges, we propose to initialize $i$ and $j$ as a two-vertices community, where edge $e_{ij}$ has the highest weight among all of $i$'s associated edges (randomly selecting a vertex $j$ if there are ties), and initialize all the other vertices within $i$'s adjacent communities as singleton communities.

\begin{prop} \label{prop:7}
If a new vertex $i$ has been added and all of its associated edges are connected to the same community, i.e., $c_{j}$, merging $i$ into $c_{j}$ has higher modularity gain than keeping $i$ as a singleton community.
\end{prop}

See Appendix \ref{Appendix_A_7} for the proof.


\begin{prop} \label{prop:8}
Suppose a new vertex $i$ has been added and its associated edges are connected to different communities. Let $\triangle w_{ij}$ denote the sum of the edge weights of vertex $i$'s associated edges that are connected to community $c_{j}$. Given two communities $c_{p}$ and $c_{q}$, if $\triangle w_{ip} > \triangle w_{iq}$, merging $i$ into $c_{p}$ has more modularity gain than merging $i$ into $c_{q}$.
\end{prop}

See Appendix \ref{Appendix_A_8} for the proof.

\subsubsection{Vertex Deletion (VD)}
In this scenario, an old vertex $i$ and its associated edges are deleted. On one hand, if $i$ has no associated edge, deleting $i$ has no influence on the rest of the network, and hence, we should keep the community structure unchanged. On the other hand, if $i$ has too many associated edges, deleting $i$ might cause its community and its neighbor communities being broken into smaller communities and potentially being merged into other communities. To handle this case, we propose to initialize all the vertices within $c_{i}$ and $i$'s neighbor communities as singleton communities.

\begin{algorithm}[t]
\caption{DynaMo Initialization (Init)}\label{alg:DynaMoInitialization}
\LinesNumbered
\KwIn{$V^{(t+1)}$, $E^{(t+1)}$, $V^{(t)}$, $E^{(t)}$, $C^{(t)}$.}
\KwOut{$\triangle C_{1}$, $\triangle C_{2}$.}

$\triangle E \gets A \ set \ of \ edges \ changed \ from \ E^{(t)} \ to \ E^{(t+1)}$\;
$\triangle V_{add} \gets V^{(t+1)} \backslash V^{(t)}$; $\triangle V_{del} \gets V^{(t)} \backslash V^{(t+1)}$\;
$\triangle C_{1} \gets \O$; $\triangle C_{2} \gets \O$\;

\For {$e_{ij} \in \triangle E$}{
    \For {$k \in \{i, j\}$}{
        \If {$k \in \triangle V_{del}$} {
            $\triangle C_{1} \gets \triangle C_{1} \cup \{c_{k}\}$\;
            \For {$e_{kl} \in E^{(t)}$}{
                $\triangle C_{1} \gets \triangle C_{1} \cup \{c_{l}\}$\;
            }
        }
        \If {$k \in \triangle V_{add}$} {
            $\triangle C_{1} \gets \triangle C_{1} \cup \{c_{k}\}$\;
            $w_{max}=0$; $c \gets \O$\;
            \For {$e_{kl} \in E^{(t+1)}$}{
                $\triangle C_{1} \gets \triangle C_{1} \cup \{c_{l}\}$\;
                \If {$w_{kl}>w_{max}$}{
                    $w_{max}=w_{kl}$; $c \gets \{k, l\}$\;
                }
            }
            $\triangle C_{2} \gets \triangle C_{2} \cup \{c\}$\;
        }
    }
    \If {$i, j \notin \triangle V_{del} \cup \triangle V_{add}$}{
        \If {$e_{ij} \notin E^{(t+1)}$ \textbf{or} $w^{t}_{ij} > w^{t+1}_{ij}$}{
            \If {$c_{i} = c_{j}$}{
                $\triangle C_{1} \gets \triangle C_{1} \cup \{c_{i}\}$\;
                \For {$k \in \{i, j\}$}{
                    \For {$e_{kl} \in E^{(t)}$}{
                        $\triangle C_{1} \gets \triangle C_{1} \cup \{c_{l}\}$\;
                    }
                }
            }
        }
        \If {$e_{ij} \notin E^{(t)}$ \textbf{or} $w^{t}_{ij} < w^{t+1}_{ij}$}{
            \If {$c_{i} = c_{j}$}{
                 $\triangle C_{1} \gets \triangle C_{1} \cup \{c_{i}\}$; $c \gets \{i, j\}$\;
                 $\triangle C_{2} \gets \triangle C_{2} \cup \{c\}$\;
            }
            \Else{
                 $\triangle w = w^{t+1}_{ij} - w^{t}_{ij}$; $c_{k} = c_{i} \cup c_{j}$\;
                 $\alpha_{2}=\alpha_{c_{i}}+\alpha_{c_{j}}-\alpha_{c_{k}}$; $\beta_{2}=\beta_{c_{i}} + \beta_{c_{j}}$\;
                 $\delta_{1}=2m - \alpha_{2} - \beta_{2}$; $\delta_{2}=m\alpha_{2}+\beta_{c_{i}}\beta_{c_{j}}$\;
                \If {$2 \triangle w + \delta_{1} > \sqrt{\delta_{1}^{2} + 4\delta_{2}}$}{
                     $\triangle C_{1} \gets \triangle C_{1} \cup \{c_{i}, c_{j}\}$; $c \gets \{i, j\}$\;
                     $\triangle C_{2} \gets \triangle C_{2} \cup \{c\}$\;
                }
            }
        }
    }
}

\Return $\triangle C_{1}$, $\triangle C_{2}$.
\end{algorithm}

\begin{algorithm}[t]
\caption{DynaMo}\label{alg:DynaMo}
\LinesNumbered
\KwIn{$G^{(t+1)}$, $G^{(t)}$, $C^{(t)}$.}
\KwOut{$C^{(t+1)}$.}

 $\triangle C_{1}$, $\triangle C_{2}$ $\gets$ \textbf{Init}($V^{(t+1)}$, $E^{(t+1)}$, $V^{(t)}$, $E^{(t)}$, $C^{(t)}$)\;
 $C^{(t+1)} \gets C^{(t)}$\;

\For {$c_{i} \in \triangle C_{1}$}{
     $C^{(t+1)} \gets C^{(t+1)} \backslash \{c_{i}\}$\;
        \For {$k \in c_{i}$}{
             \text{Create singleton community:} $c_{k} \gets \{k\}$\;
             $C^{(t+1)} \gets C^{(t+1)} \cup \{c_{k}\}$\;
        }
}

\For {$c=\{i, j\} \in \triangle C_{2}$}{
    \text{Create two-vertices community:} $c_{k} \gets \{i, j\}$\;
    $C^{(t+1)} \gets (C^{(t+1)} \backslash \{c_{i}, c_{j}\}) \cup \{c_{k}\}$\;
}

$C^{(t+1)}$ $\gets$ \textbf{Louvain}($C^{(t+1)}$, $G^{(t+1)}$)\;

\Return $C^{(t+1)}$.
\end{algorithm}

\subsection{Implementation and Analysis}
\subsubsection{Implementation}
Algorithm~\ref{alg:DynaMoInitialization} presents the DynaMo Initialization, where we implement the operation of each type of incremental network change to initialize the intermediate community structure towards maximizing the modularity. The input contains the current network $G^{(t+1)}$, the previous network $G^{(t)}$ and the previous community structure $C^{(t)}$. The output contains two set of communities, $\triangle C_{1}$ and $\triangle C_{2}$, that will be modified to initialize the intermediate community structure at the beginning of the second phase. $\triangle C_{1}$ contains a set of communities in $C^{(t)}$ to be separated into singleton communities, and $\triangle C_{2}$ contains a set of two-vertices communities to be created.
Algorithm~\ref{alg:DynaMo} presents the second phase, where the last two steps of Louvain algorithm is applied on the initialized intermediate community structure of $G^{(t+1)}$.

Most of the operations in Algorithm~\ref{alg:DynaMoInitialization} are theoretically guaranteed by our propositions described in Section~\ref{sec:DynaMoAlgorithm} to maximize the modularity, while some of the operations are heuristically designed for the sake of the efficiency.
For instance, according to Proposition~\ref{prop:1}, Remark~\ref{remark:1} and Proposition~\ref{prop:2}, given ICEA/WI between vertices $i$ and $j$, we initialize $i$ and $j$ as a two-vertices community to incrementally maximize the modularity, and initialize all the other vertices in $c_{i}$ as singleton communities to take all the influenced vertices into consideration carefully while maintaining the algorithm efficiency (lines 26-28).
According to Proposition~\ref{prop:4} and Proposition~\ref{prop:5}, we use a designed threshold condition (lines 30-33) to determine the operation of given CCEA/WI. If the condition is true, we use the same operation of ICEA/WI to tackle CCEA/WI (lines 33-35). Otherwise, we keep the community structure unchanged to incrementally maximize the modularity.
According to Proposition~\ref{prop:3} and the analysis in Section~\ref{sec:DynaMoAlgorithm_EDWD}, given ICED/WD, we initialize all the potentially influenced vertices as singleton communities to maintain a trade-off between the effectiveness and efficiency (lines 18-24). According to Proposition~\ref{prop:6}, given CCED/WD, we keep the community structure unchanged to maximize the local modularity gain.
According to Proposition~\ref{prop:7} and Proposition~\ref{prop:8}, given new vertex $i$ and its associated edges, we initialize $i$ and its most closely connected neighbor vertex as a two-vertices community (lines 12, 15-17), and initialize all the potentially influenced vertices as singleton communities (lines 10-16). After deleting vertex $i$, we heuristically initialize all the vertices within $c_{i}$ and $i$'s neighbor communities as singleton communities (lines 6-9). To summarize, initializing $\triangle C_{2}$ aims to incrementally maximize the modularity with certain theoretical guarantees, and initializing $\triangle C_{1}$ aims to heuristically maximize the modularity (by Algorithm~\ref{alg:DynaMo}) while maintaining the algorithm efficiency.

\subsubsection{Time Complexity Analysis} \label{sec:DynaMo_TimecomplexityAnalysis}
The computation of our algorithm tackling one network snapshot 
comes from two parts: (a) the initialization, and (b) the last two steps of Louvain algorithm. In the initialization, different network changes trigger different operations, thus resulting in different computation time. For instance, if one network change is ICEA/WI (i.e., $e_{ij}$, $c_{i}=c_{j}$), our algorithm (line 26-28) will add $c_{i}$ into $\triangle C_{1}$, 
and add $c=\{i,j\}$ into $\triangle C_{2}$. 
The time complexity of both operations are $O(1)$, thus, the time complexity to deal with single change of ICEA/WI is $O(1)$. Similarly, the time complexities to deal with single change of CCEA/WI (line 29-35) and CCED/WD (no operation needed) are also $O(1)$. To deal with single change of ICED/WD, VA or VD, our algorithm runs through the set of neighbor vertices of the changed edge, and thus, result in $O(\frac{|E|}{|V|})$ 
time complexity. Furthermore, as shown in Algorithm~\ref{alg:DynaMoInitialization}, each network snapshot usually has multiple network changes. Since the number of network changes is proportional to $\triangle E$, the overall time complexity of the initialization 
is $O(|\triangle E|)$ or $O(|\triangle E| \cdot \frac{|E|}{|V|})$.

The time complexity of the original Louvain algorithm is $O(|E|)$. However, compared with the Louvain algorithm initialization, our algorithm considers the historical information and designs an initialization phase to reduce the number of edges left for the second phase analysis as much as possible. Thus, 
the time complexity of the second phase of our algorithm is $O(|E|^{*})$, where $|E|^{*} \ll |E|$. Hence, the overall best case time complexity of our algorithm 
is $O(|\triangle E| +|E|^{*})$, and the worst case is $O(|\triangle E| \cdot \frac{|E|}{|V|}+|E|^{*})$.

\section{Experimental Evaluation} \label{sec:ExperimentalEvaluation}
\subsection{Experiment Environment}
All the experiments were conducted on a PC with an Intel Xeon Gold 6148 Processor, 128GB RAM, running 64-bit Ubuntu 18.04 LTS operating system. All the algorithms and experiments are implemented using Java with JDK 8.

\subsection{Baseline Approaches}
We compare DynaMo with {\bf Louvain} (Section~\ref{sec:Preliminaries_Louvain}), and 5 dynamic algorithms: (i) {\bf Batch \cite{chong2013incremental}:} a batch-based incremental modularity optimization algorithm; (ii) {\bf GreMod \cite{shang2014real}:} a rule-based incremental algorithm that performs predetermined operations on edge additions; (iii) {\bf QCA \cite{nguyen2011adaptive}:} a rule-based incremental algorithm that updates the community structures according to predefined rules of vertex/edge additions/deletions; (iv) {\bf LBTR \cite{shang2016targeted}:} a learning-based algorithm that uses classifiers to update community assignments. We use Support Vector Machine (SVM) and Logistic Regression (LR) as the classifiers, namely {\bf LBTR-SVM} and {\bf LBTR-LR}.

\begin{table*}[!t]
\footnotesize
\captionsetup{font=footnotesize}
\caption{Description of the real-world dynamic networks [Notations: $|V|$ ($|E|$): $\#$ of unique vertices (edges); $\mathbb{E}[|\triangle V|]$ ($\mathbb{E}[|\triangle E|]$): avg. $\#$ of vertices (edges) changed per network snapshots; $\#$ of snapshots: total number of consecutive network snapshots; time-interval: period of time between two consecutive network snapshots; time-span: total time spanning of each network dataset].}
\label{table:EvolvingNetworks}
\centering
\begin{tabular}{c|c|c|c|c|c|c|c|c|c}
\hline
 \bfseries networks & $\mathbf{|V|}$ & $\mathbf{\mathbb{E}[|\triangle V|]}$ & \bfseries vertex-type & $\mathbf{|E|}$  & $\mathbf{\mathbb{E}[|\triangle E|]}$ & \bfseries edge-type & \bfseries $\#$ of snapshots & \bfseries time-interval & \bfseries time-span\\
\hline
 \bfseries Cit-HepPh  & 30,501 & 6,460  & author  & 346,742  & 11,127 & co-citation & 31 & 4 months &  124 months\\
 \hline
 \bfseries Cit-HepTh  & 7,577 & 1,253  & author  & 51,089  & 2,042 & co-citation & 25 & 5 months &  125 months\\
 \hline
 \bfseries DBLP  & 1,411,321  & 122,731 &  author & 5,928,285 & 191,233 & co-authorship & 31 & 2 years &  62 years\\
 \hline
 \bfseries Facebook  & 59,302 & 12,765 & user  & 592,406  & 20,943 & friendship & 28 & 1 month & 28 months \\
 \hline
 \bfseries Flickr  & 780,079 & 93,253 & user  & 4,407,259 & 168,977 & follow & 24 & 3 days & 72 days \\
 \hline
 \bfseries YouTube  & 3,160,656 & 91,954  & user  &  7,211,498 & 175,303 & subscription & 33 & 5 days & 165 days \\
\hline
\end{tabular}
\end{table*}

\subsection{Experiment Datasets}
We conduct our experiments on two categories of networks: real-world networks (ground-truth is unknown), and synthetic networks (ground-truth is known).
\subsubsection{Real-world Dynamic Networks}
As shown in Table~\ref{table:EvolvingNetworks}, six real-world networks are used in our experiments. (i) {\bf Cit-HepPh (Cit-HepTh) \cite{leskovec2005graphs}} contains the citation network of high-energy physics phenomenology (theory) papers from 1993 to 2003. (ii) {\bf DBLP \cite{bader2013graph}} contains a co-authorship network of computer science papers ranging from 1954 to 2015, where each author is represented as a vertex and co-authors are linked by an edge. (iii) {\bf Facebook \cite{viswanath2009evolution}} contains the user friendship establishment information from about 52\% of Facebook users in New Orleans area, spanning from September 26th, 2006 to January 22nd, 2009. In this network, each vertex represents a Facebook user, and each edge represents an user-to-user friendship establishment link that contains a timestamp representing the time of friendship establishment. (iv) {\bf Flickr \cite{mislove2008growth}} was obtained on January 9th, 2007, and contains over 1.8 million users and 22 million links, and each link has a timestamp that represents the time of the following link establishment. We select a sub-network, where all the user-to-user following links were established from March 6th, 2007 to May 15th, 2007. (v) {\bf YouTube \cite{mislove2007measurement}} was obtained on January 15th, 2007 and consists of over 1.1 million users and 4.9 million links, and each link has a timestamp that represents the time of the subscribing link establishment. We select a sub-network, where all the user-to-user subscribing links were established from February 2nd, 2007 to July 23rd, 2007.

\subsubsection{Synthetic Dynamic Networks}
We use RDyn \cite{rossetti2017graph}, a benchmark model focusing on community changes in dynamic networks, to generate synthetic networks and their ground-truth communities. It allows us to specify different parameters, such as the number of vertices ($N$), the number of time points ($T$), the maximum number of community change events (e.g., splitting or merging) per time point ($M$), etc.. We use various combinations of $N$, $T$ and $M$ to generate synthetic networks, where $N$ $\in$ $\{200,$ $400,$ $600,$ $800,$ $1000\}$, $T$ $\in$ $\{25,$ $50,$ $75,$ $100,$ $125\}$, $M$ $\in$ $\{1,$ $2,$ $3,$ $4\}$ and all the other parameters set by default values. For each parameter combination (out of 100 combinations in total), we randomly generate 100 synthetic networks, resulting in 10,000 synthetic networks in total.

\subsection{Experimental Procedure}
For each real-world network, we apply Louvain algorithm on its initial snapshot to obtain its initial community structure (Section~\ref{sec:DynaMo_MethodologyOverview}). For each synthetic network, we use the ground-truth communities of its initial snapshot as its initial community structure. For the rest of snapshots of real-world and synthetic networks, the dynamic algorithms only use the initial community structure and the network changes between two consecutive snapshots to update the new community structures, while the static algorithm will be applied on the whole network of each snapshot. All experiments are performed for 200 times to obtain the average results.

\begin{figure*}[tb]
\captionsetup{font=footnotesize}
        \centering
        \begin{subfigure}[b]{0.32\textwidth}
                \includegraphics[width=\textwidth]{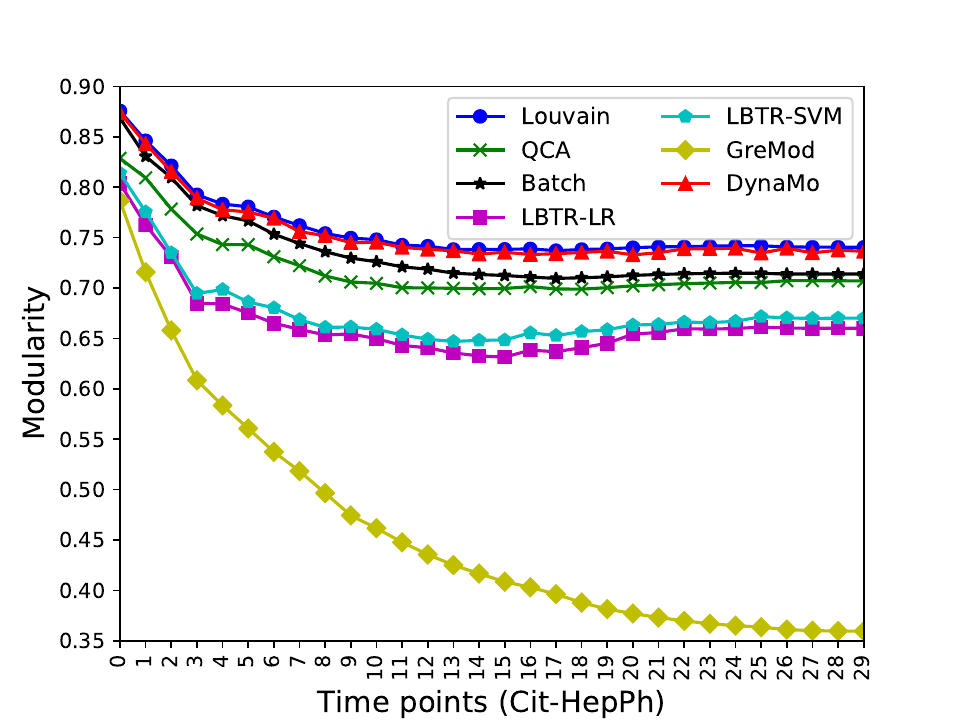}
                \caption{}
                \label{fig:Cit-HepPh_MOD_CP}
        \end{subfigure}%
        ~ 
        \begin{subfigure}[b]{0.32\textwidth}
                \includegraphics[width=\textwidth]{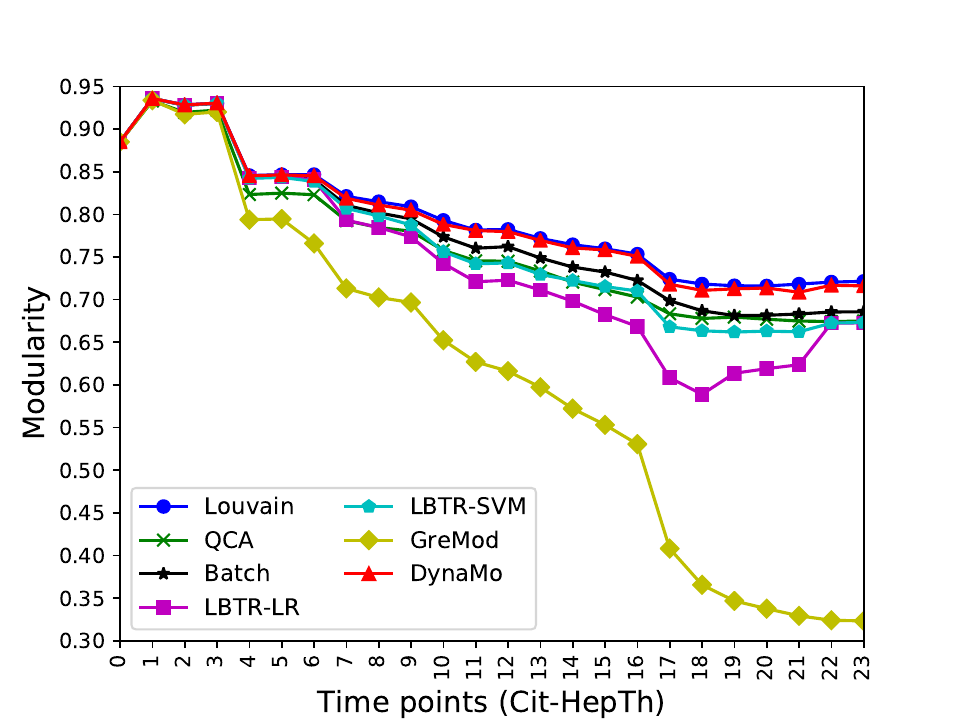}
                \caption{}
                \label{fig:Cit-HepTh_MOD_CP}
        \end{subfigure}
        ~ 
        \begin{subfigure}[b]{0.32\textwidth}
                \includegraphics[width=\textwidth]{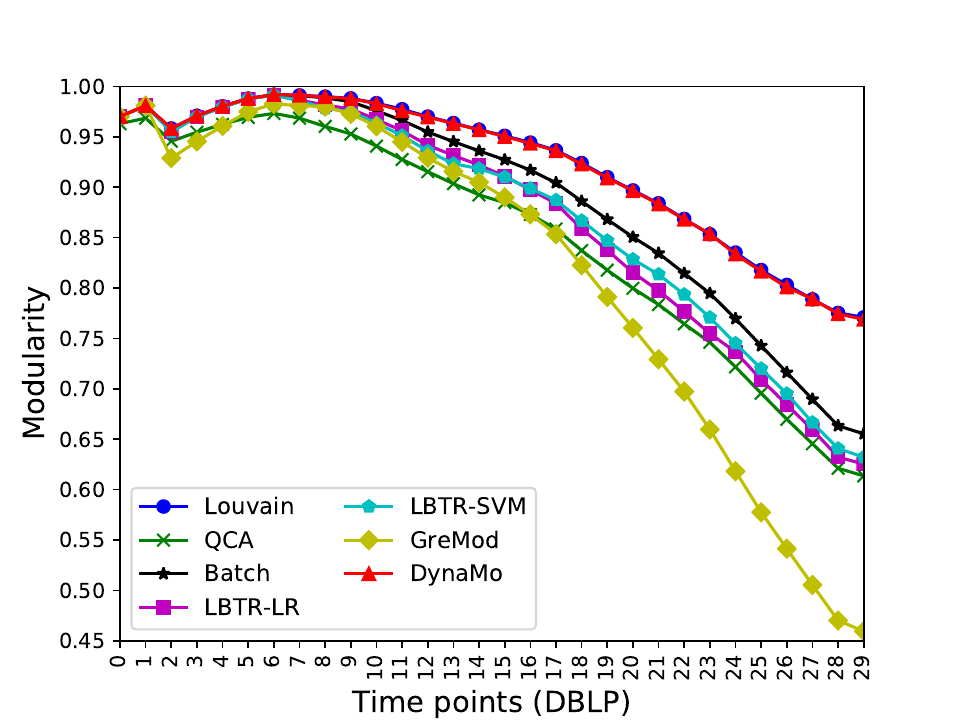}
                \caption{}
                \label{fig:DBLP_MOD_CP}
        \end{subfigure}
        ~ 
        \begin{subfigure}[b]{0.32\textwidth}
                \includegraphics[width=\textwidth]{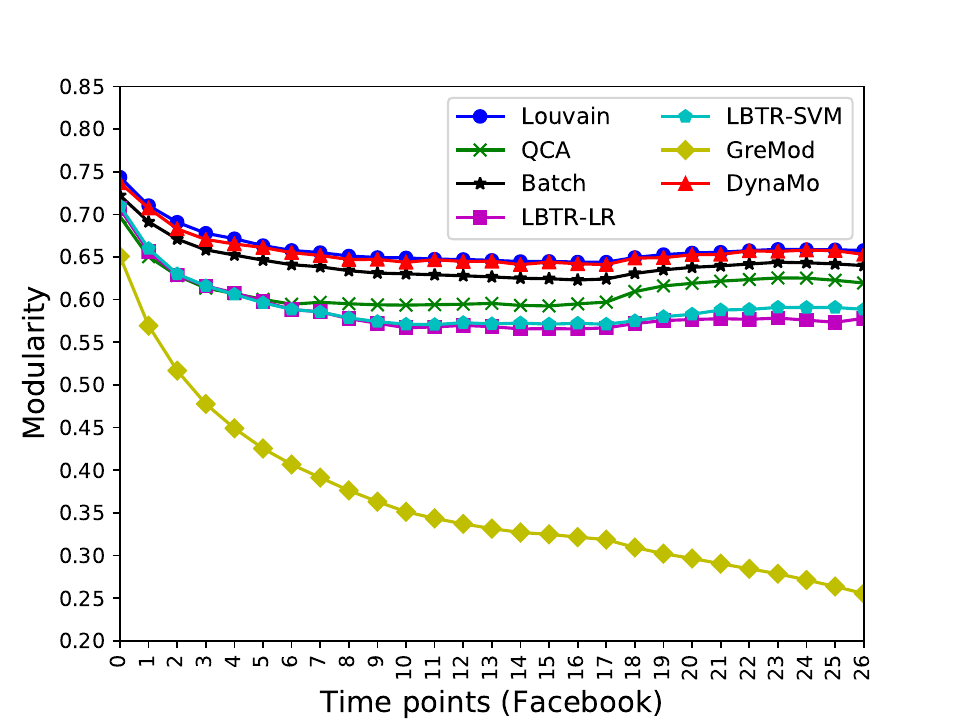}
                \caption{}
                \label{fig:Facebook_MOD_CP}
        \end{subfigure}
        ~ 
        \begin{subfigure}[b]{0.32\textwidth}
                \includegraphics[width=\textwidth]{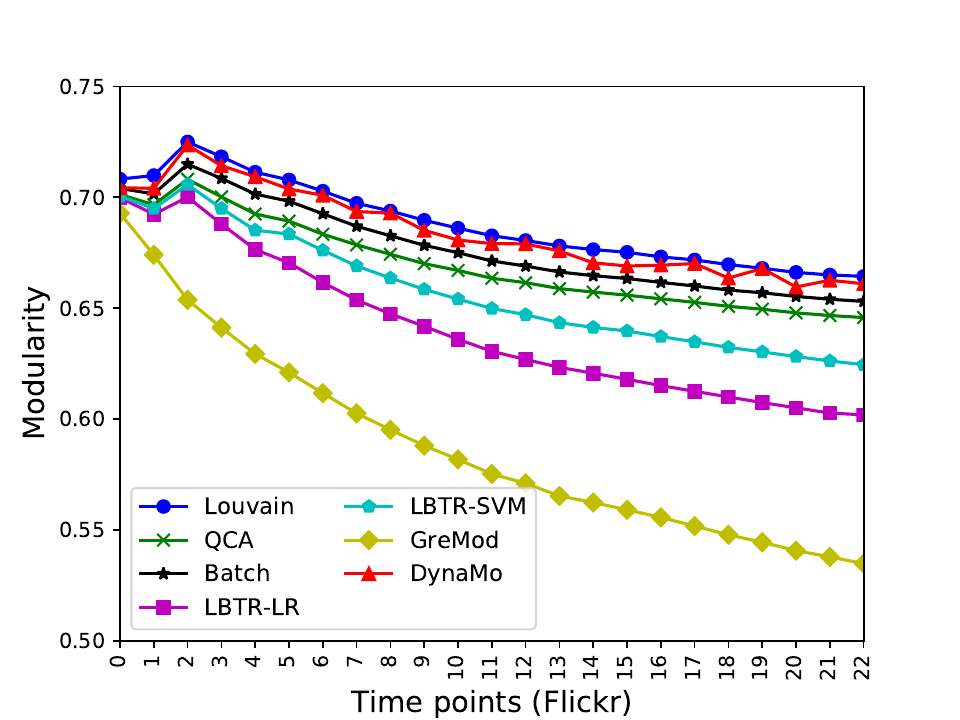}
                \caption{}
                \label{fig:Flickr_MOD_CP}
        \end{subfigure}
        ~ 
        \begin{subfigure}[b]{0.32\textwidth}
                \includegraphics[width=\textwidth]{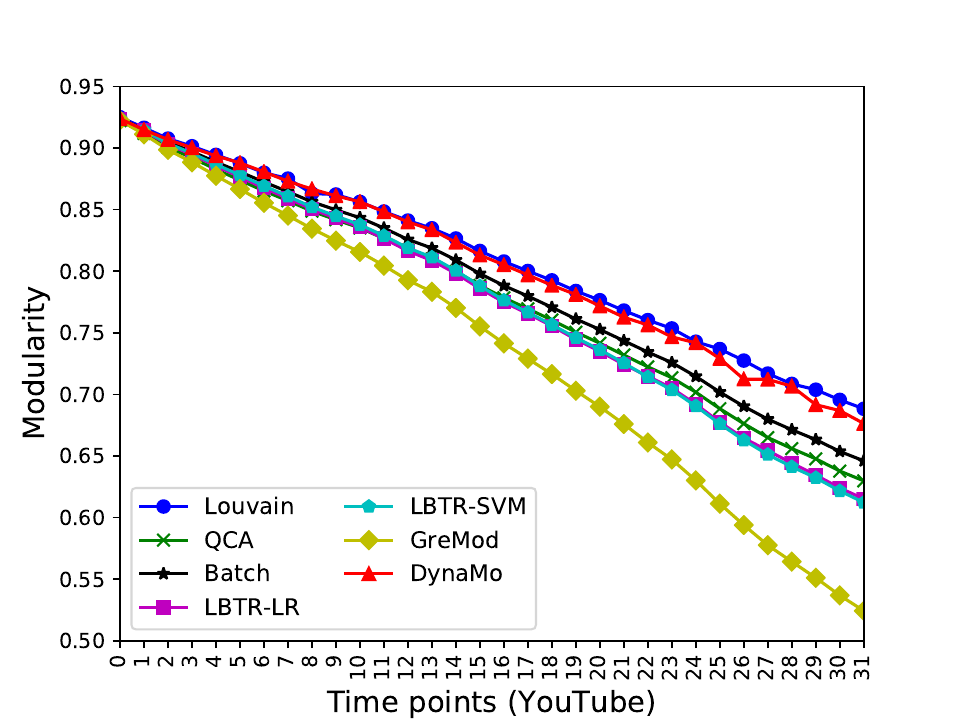}
                \caption{}
                \label{fig:YouTube_MOD_CP}
        \end{subfigure}
        \caption{
        The modularity results of real-world networks. (a) Cit-HepPh. (b) Cit-HepTh. (c) DBLP. (d) Facebook. (e) Flickr. (f) YouTube.}
        \label{fig:Modularity}
\end{figure*}

\begin{figure*}[tb]
\captionsetup{font=footnotesize}
        \centering
        \begin{subfigure}[b]{0.32\textwidth}
                \includegraphics[width=\textwidth]{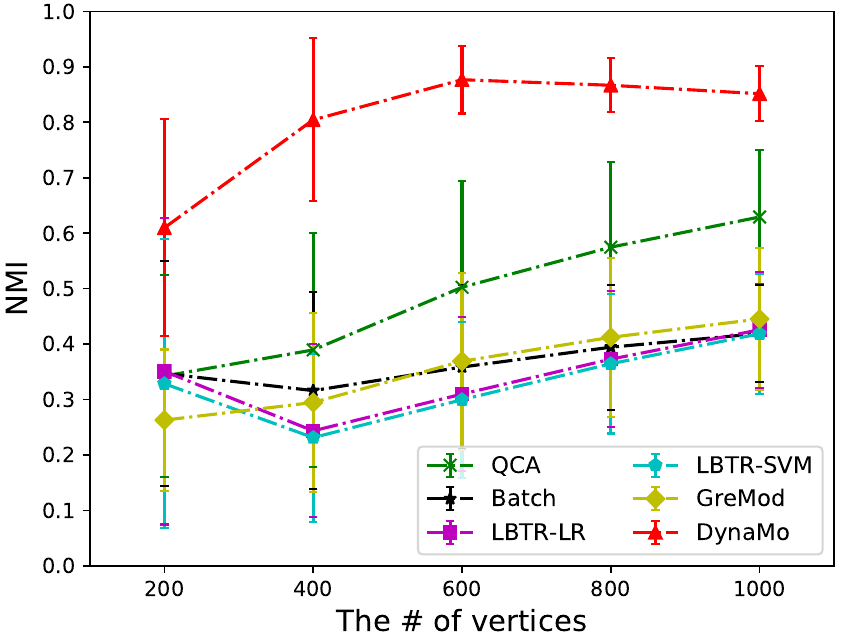}
                \caption{}
                \label{fig:NMI_vertices_1_CP}
        \end{subfigure}%
        ~ 
        \begin{subfigure}[b]{0.32\textwidth}
                \includegraphics[width=\textwidth]{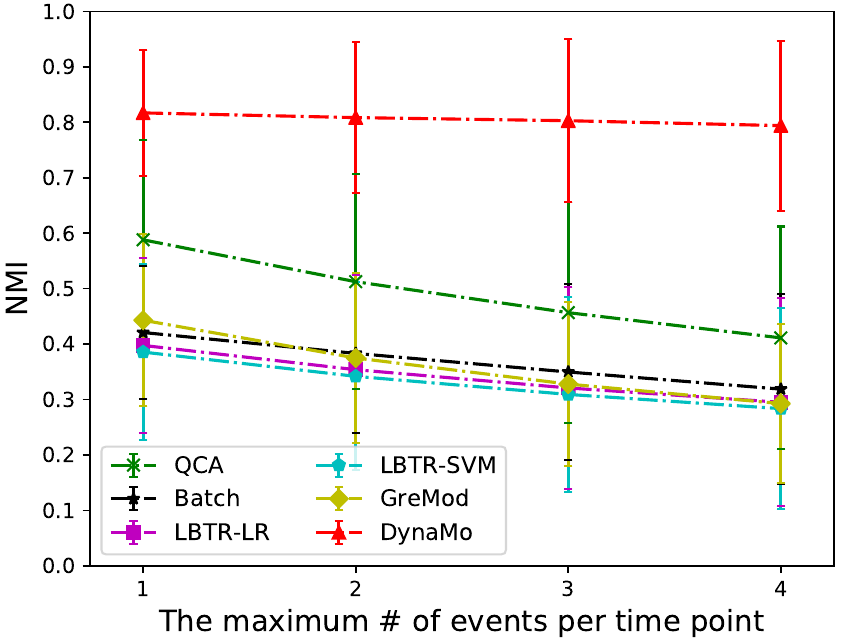}
                \caption{}
                \label{fig:NMI_events_1_CP}
        \end{subfigure}
        ~ 
        \begin{subfigure}[b]{0.32\textwidth}
                \includegraphics[width=\textwidth]{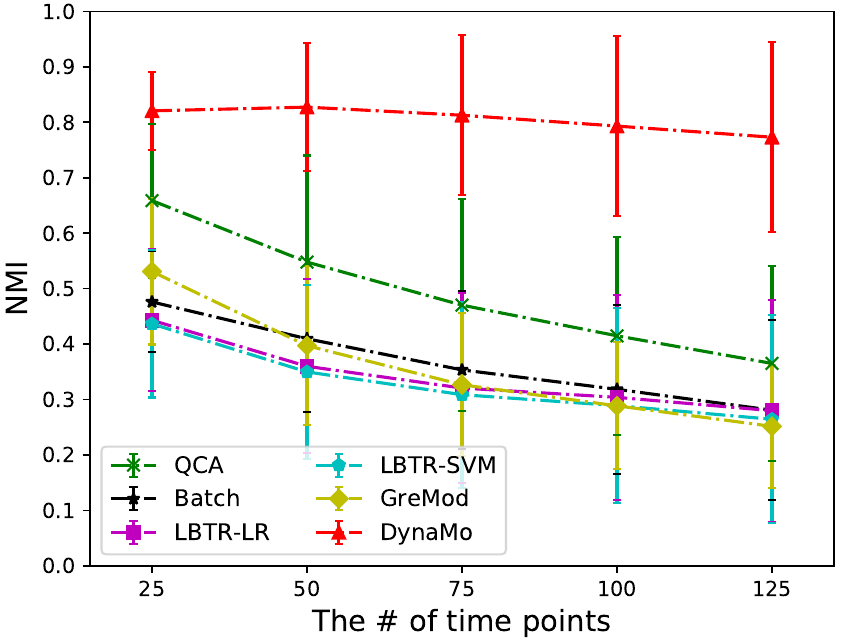}
                \caption{}
                \label{fig:NMI_time_1_CP}
        \end{subfigure}
        \caption{
        The NMI results of synthetic networks. (a) The $\#$ of vertices. (b) The maximum $\#$ of events per time point. (c) The $\#$ of time points.}
        \label{fig:NMI}
\end{figure*}

\begin{figure*}[tb]
\captionsetup{font=footnotesize}
        \centering
        \begin{subfigure}[b]{0.32\textwidth}
                \includegraphics[width=\textwidth]{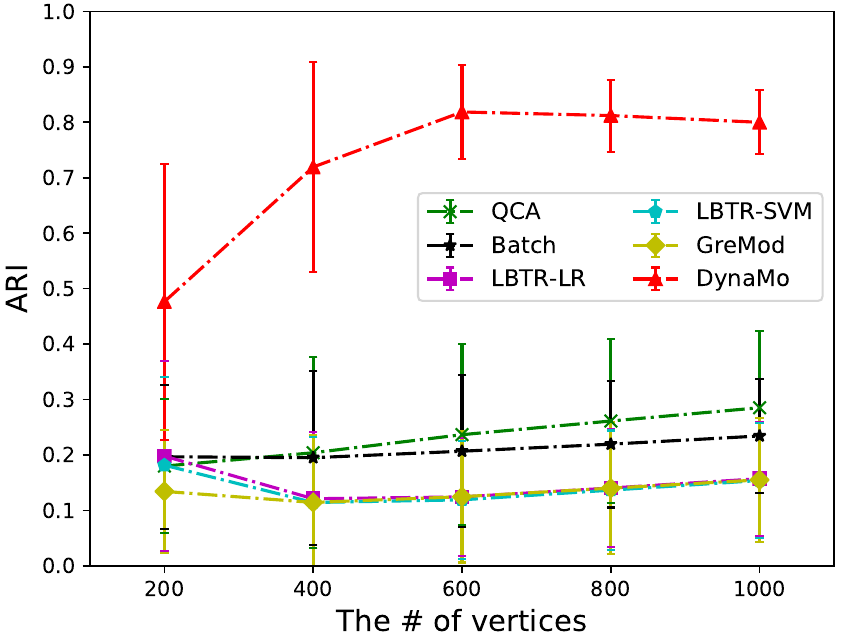}
                \caption{}
                \label{fig:ARI_vertices_1_CP}
        \end{subfigure}%
        ~ 
        \begin{subfigure}[b]{0.32\textwidth}
                \includegraphics[width=\textwidth]{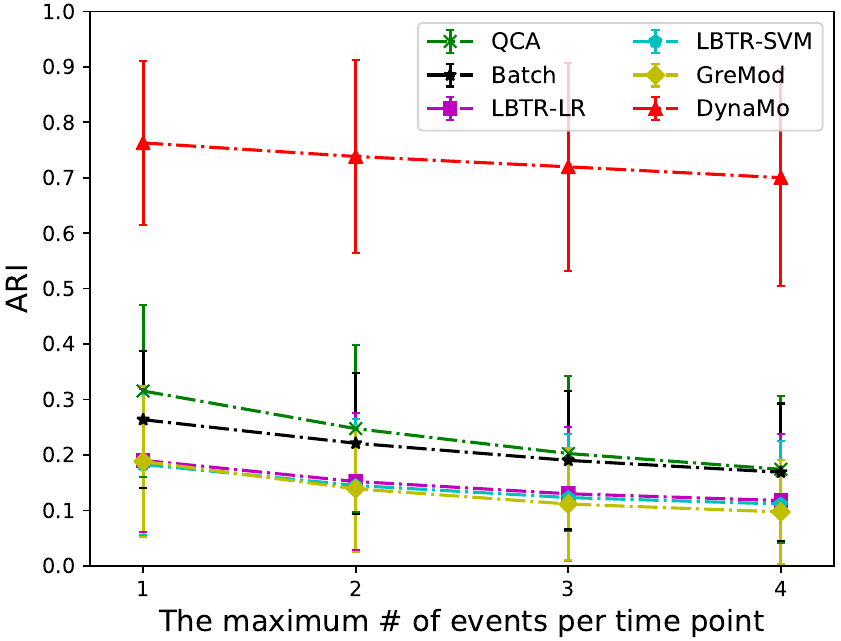}
                \caption{}
                \label{fig:ARI_events_1_CP}
        \end{subfigure}
        ~ 
        \begin{subfigure}[b]{0.32\textwidth}
                \includegraphics[width=\textwidth]{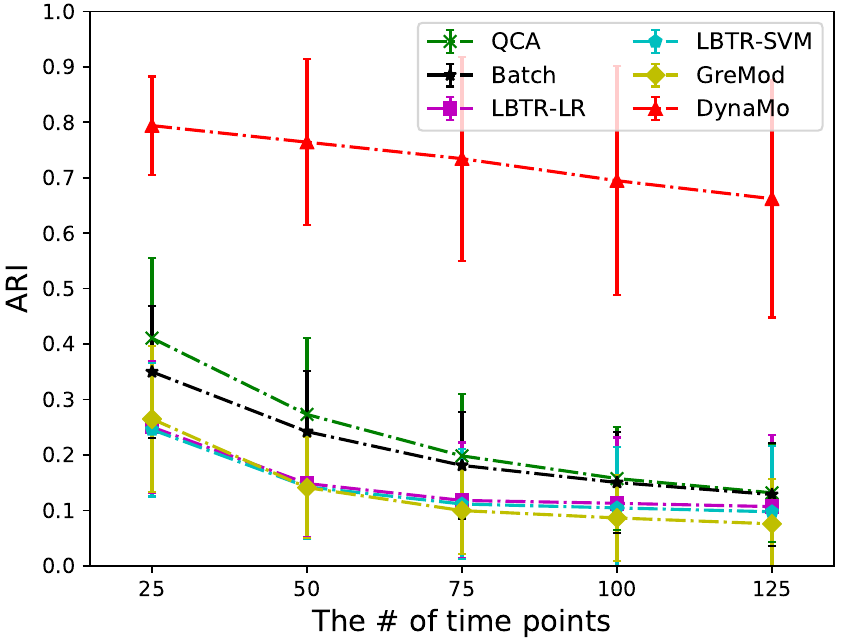}
                \caption{}
                \label{fig:ARI_time_1_CP}
        \end{subfigure}
        \caption{
        The ARI results of synthetic networks. (a) The $\#$ of vertices. (b) The maximum $\#$ of events per time point. (c) The $\#$ of time points.}
        \label{fig:ARI}
\end{figure*}

\begin{figure*}[tb]
\captionsetup{font=footnotesize}
        \centering
        \begin{subfigure}[b]{0.32\textwidth}
                \includegraphics[width=\textwidth]{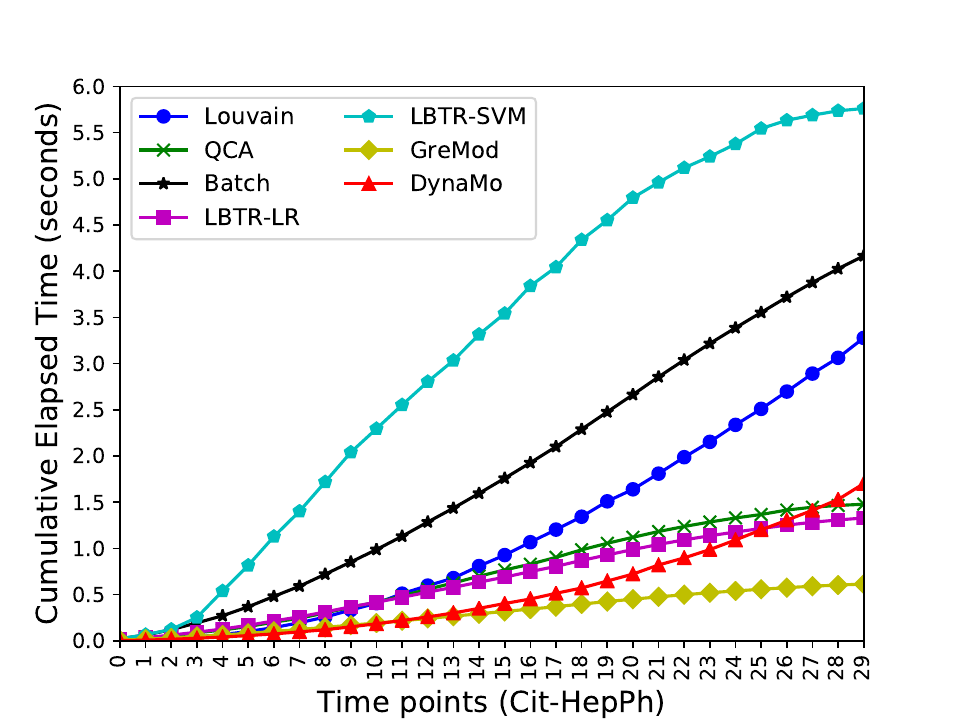}
                \caption{}
                \label{fig:Cit-HepPh_CET_CP}
        \end{subfigure}%
        ~ 
        \begin{subfigure}[b]{0.32\textwidth}
                \includegraphics[width=\textwidth]{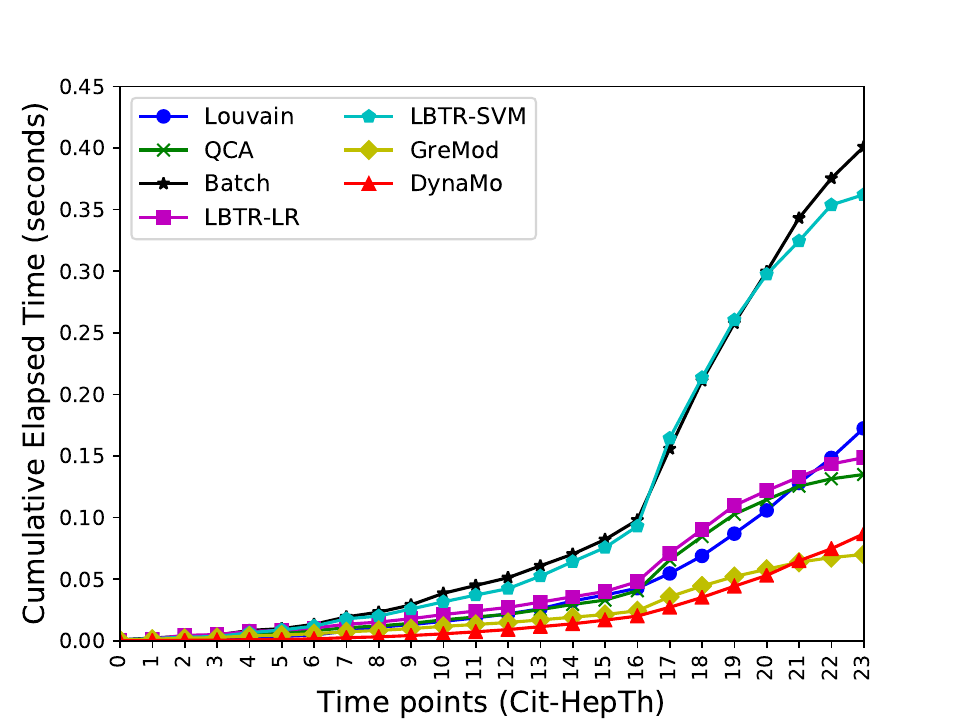}
                \caption{}
                \label{fig:Cit-HepTh_CET_CP}
        \end{subfigure}
        ~ 
        \begin{subfigure}[b]{0.32\textwidth}
                \includegraphics[width=\textwidth]{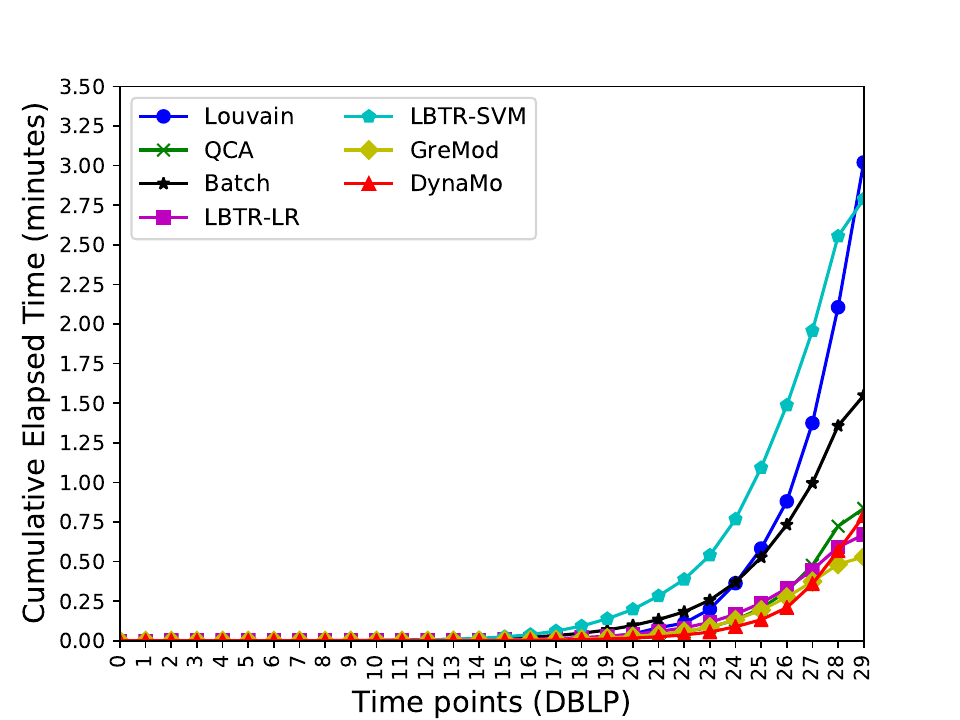}
                \caption{}
                \label{fig:dblp_coauthorship_CET_CP}
        \end{subfigure}
        ~ 
        \begin{subfigure}[b]{0.32\textwidth}
                \includegraphics[width=\textwidth]{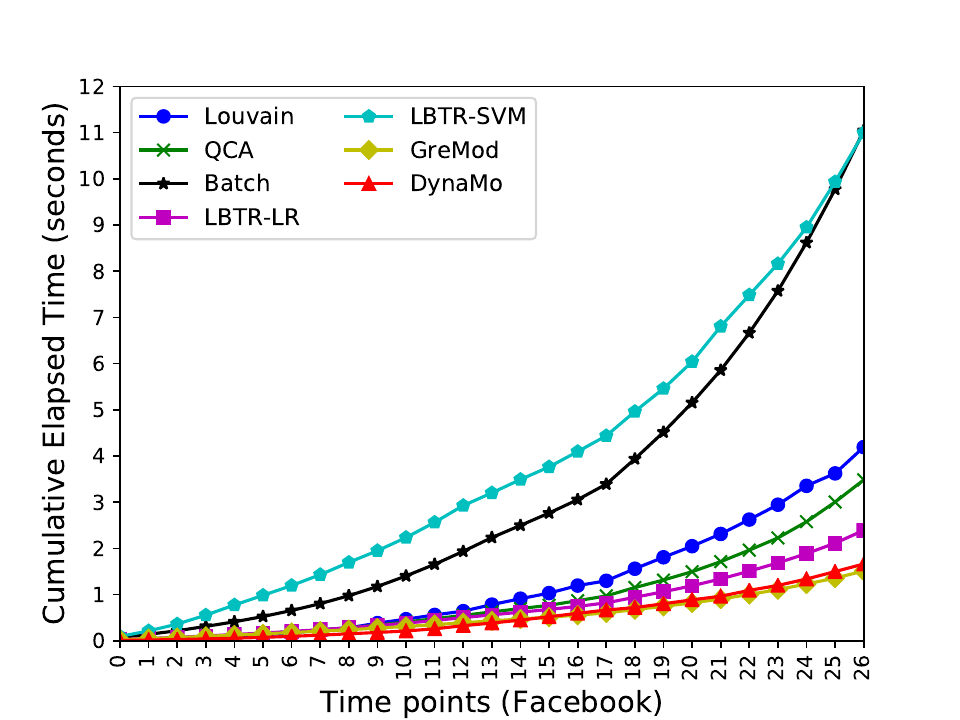}
                \caption{}
                \label{fig:facebook_CET_CP}
        \end{subfigure}
        ~ 
        \begin{subfigure}[b]{0.32\textwidth}
                \includegraphics[width=\textwidth]{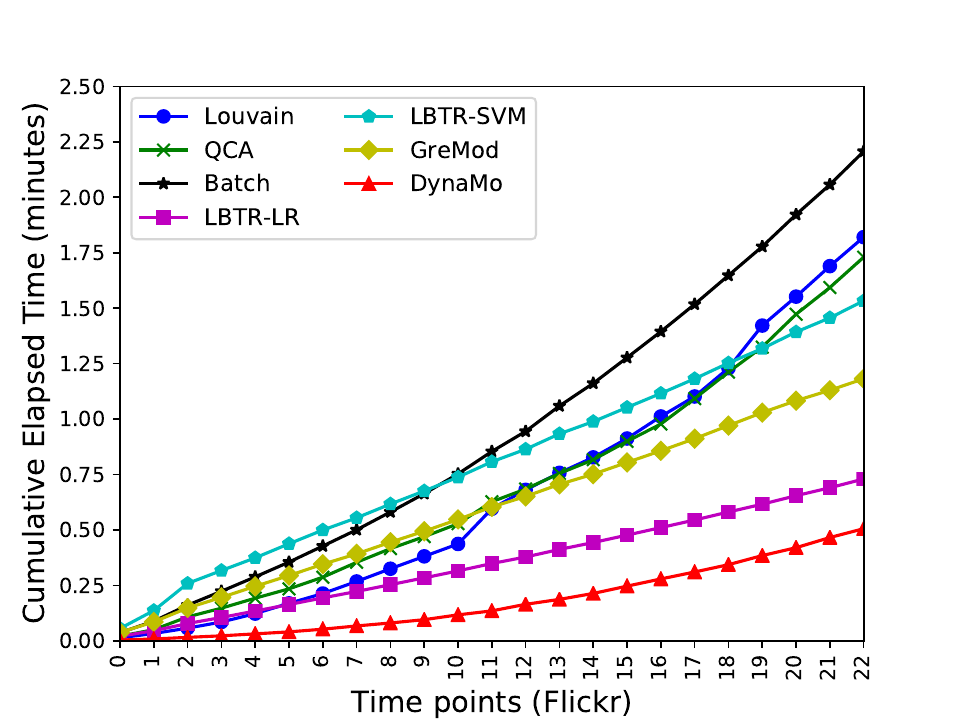}
                \caption{}
                \label{fig:flickr_CET_CP}
        \end{subfigure}
        ~ 
        \begin{subfigure}[b]{0.32\textwidth}
                \includegraphics[width=\textwidth]{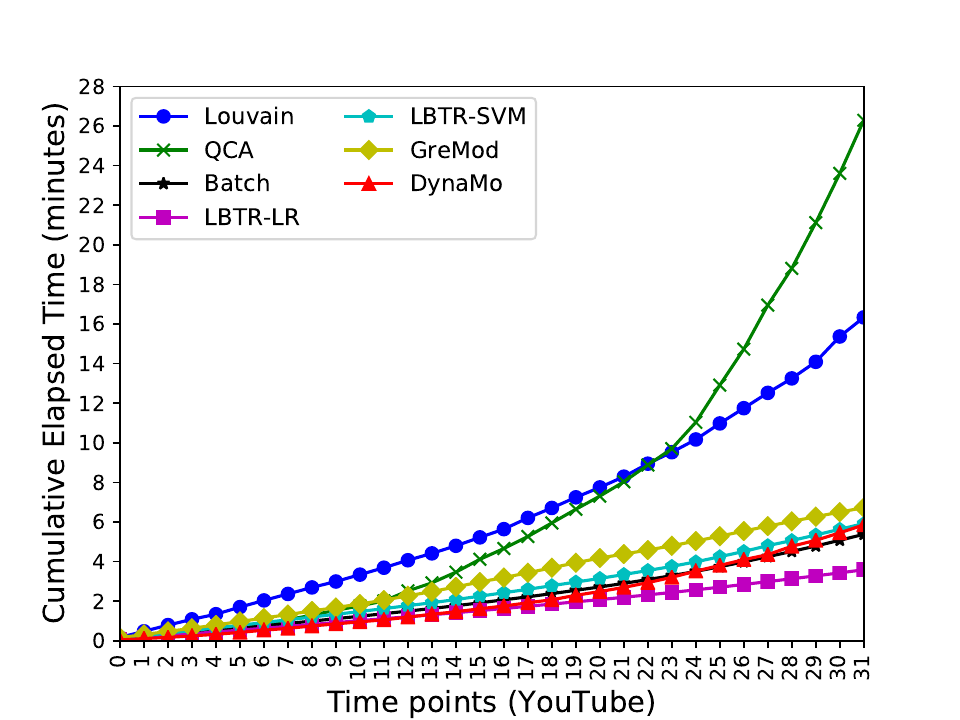}
                \caption{}
                \label{fig:youtube_CET_CP}
        \end{subfigure}
        \caption{
        The cumulative elapsed time results of real world networks. (a) Cit-HepPh. (b) Cit-HepTh. (c) DBLP. (d) Facebook. (e) Flickr. (f) YouTube.}
        \label{fig:CET}
\end{figure*}

\subsection{Effectiveness Analysis} \label{sec:ExperimentalEvaluation_Effectiveness}
\subsubsection{Effectiveness Metrics}
We evaluate the effectiveness of the community detection algorithms using three metrics: modularity, Normalized Mutual Information (NMI) and Adjusted Rand Index (ARI). Modularity (Section~\ref{sec:Preliminaries_Modularity}) is designed to measure the strength of dividing a network into communities, and does not require the ground-truth information. Hence, we use modularity to evaluate the results of the real-world networks. NMI and ARI are designed to measure the similarities between the community structure obtained from the experiments and that of the ground-truth, which are used to evaluate the results of the synthetic networks.

Let $C_{t}$ denote the ground-truth community division, and $C_{r}$ denote the experiment result. NMI is defined as follows:
\begin{equation}
    \begin{split} \label{eq:nmi}
        NMI(C_{t}, C_{r})=\frac{2 \times I(C_{t}; C_{r})}{[H(C_{t})+H(C_{r})]}
    \end{split}
\end{equation}
where $H(C_{r})$ is the entropy of $C_{r}$, and $I(C_{t}; C_{r})$ is the mutual information between $C_{t}$ and $C_{r}$. NMI ranges from 0 to 1. NMI closing to 1 indicates $C_{r}$ is similar to $C_{t}$, while closing to 0 means $C_{r}$ is random compared with $C_{t}$.

Let $a$ be the number of pairs of vertices in the same community in both $C_{t}$ and $C_{r}$, $b$ be the number of pairs of vertices in the same community in $C_{t}$ and in different communities in $C_{r}$, $c$ be the number of pairs of vertices in different communities in $C_{t}$ and in the same community in $C_{r}$, $d$ be the number of pairs of vertices in different communities in both $C_{t}$ and $C_{r}$. ARI is defined as follows:
\begin{equation}
    \begin{split} \label{eq:ari}
        ARI(C_{t}, C_{r})=\frac{2(ad-bc)}{b^{2}+c^{2}+2ad+(a+d)(b+c)}
    \end{split}
\end{equation}
where its upper bound is 1, and the higher, the better.

\subsubsection{Experimental Results}
Fig.~\ref{fig:Modularity} shows the modularity results of 7 algorithms running on 6 real-world networks, respectively. We observe that DynaMo consistently outperforms all the other dynamic algorithms in terms of modularity. Compared with the runner-up algorithm (Batch), DynaMo obtains 2.6\%, 2.2\%, 4.3\%, 2.1\%, 1.1\% and 2.2\% higher modularity averaged over all the time points, and 3.2\%, 4.4\%, 17.3\%, 2.4\%, 1.2\% and 4.7\% higher modularity on the last time point of Cit-HepPh, Cit-HepTh, DBLP, Facebook, Flickr and YouTube, respectively. Compared with Louvain, DynaMo achieves nearly identical performance, with only 0.49\%, 0.38\%, 0.06\%, 0.7\%, 0.5\% and 0.5\% lower modularity averaged over all the time points, and only 0.52\%, 0.74\%, 0.27\%, 0.46\%, 0.5\% and 1.7\% lower modularity on the last time point of Cit-HepPh, Cit-HepTh, DBLP, Facebook, Flickr and YouTube, respectively.

Fig.~\ref{fig:NMI} shows the NMI results (mean and standard deviation) of 6 dynamic algorithms running on 10,000 synthetic networks. We observe that DynaMo obtains the highest NMI value among all the dynamic algorithms regardless of any RDyn parameters. DynaMo outperforms the runner-up algorithm (QCA) by 69.1\%, 66.4\% and 70.3\% on average with the increase of the number of vertices, the maximum number of events per time point, and the number of time points, respectively, which is statistically significant according to the two-sample t-test with 95\% confidence interval. The NMI standard deviation of DynaMo is also lower than that of QCA, demonstrating the consistency of DynaMo in detecting communities of various dynamic networks. Furthermore, as the maximum number of events per time point and the number of time points increase, DynaMo has the minimum NMI value loss among all the dynamic algorithms, indicating DynaMo is more robust and consistent while detecting communities of dynamic networks that last longer and have more events per time point.

Fig.~\ref{fig:ARI} shows the ARI results (mean and standard deviation) of 6 dynamic algorithms running on 10,000 synthetic networks, which share similar patterns as the NMI results. DynaMo outperforms the runner-up algorithm (QCA) by 211.1\%, 224.5\% and 257.6\% on average with the increase of the number of vertices, the maximum number of events per time point, and the number of time points, respectively, which is statistically significant according to the two-sample t-test with 99\% confidence interval. As the number of vertices increases, the ARI standard deviation of DynaMo dramatically decreases, while as the maximum number of events per time point and the number of time points increase, the standard deviation of DynaMo slightly increases. However, even considering the standard deviation difference, DynaMo still significantly outperforms all the other dynamic algorithms.

\begin{table}[!t]
\footnotesize
\captionsetup{font=footnotesize}
\caption{A comparison of the time complexities of the competing algorithms [Notations: $n=|V|$ ($m=|E|$): $\#$ of unique vertices (edges); $\upsilon=|\triangle V|$ ($\epsilon=|\triangle E|$): $\#$ of vertices (edges) changed; $m^{*}_{b}$ ($m^{*}_{d}$): $\#$ of unique vertices (edges) after the initialization phase of Batch (DynaMo), and $m^{*}_{b} \ll m$ ($m^{*}_{d} \ll m$); $T_{LR}$ ($T_{SVM}$): the time complexity of using logistic regression (support vector machine) in LBTR].}
\label{table:TimeComplexity}
\centering
\begin{tabular}{c|c|c}
\hline
 \bfseries algorithms & \bfseries the best case & \bfseries the worst case \\
\hline
 \bfseries Louvain \cite{blondel2008fast} & $O(m)$ & $O(m)$\\
\hline
 \bfseries Batch \cite{chong2013incremental} & $O((\upsilon+\epsilon)\cdot\frac{m}{n}+m^{*}_{b})$ & $O((\upsilon+\epsilon)\cdot\frac{m}{n}+m^{*}_{b})$\\
 \hline
 \bfseries DynaMo  & $O(\epsilon+m^{*}_{d})$ & $O(\epsilon\cdot\frac{m}{n}+m^{*}_{d})$\\
 \hline
 \bfseries QCA \cite{nguyen2011adaptive} & $O(\epsilon)$ & $O(\epsilon \cdot m)$\\
 \hline
 \bfseries GreMod \cite{shang2014real} & $O(\epsilon)$ & $O(\epsilon \cdot n)$\\
 \hline
 \bfseries LBTR-LR \cite{shang2016targeted} & $O(\upsilon \cdot T_{LR})$ & $O(\upsilon \cdot T_{LR})$\\
 \hline
 \bfseries LBTR-SVM \cite{shang2016targeted} & $O(\upsilon \cdot T_{SVM})$ & $O(\upsilon \cdot T_{SVM})$\\
\hline
\end{tabular}
\end{table}

\subsection{Efficiency Analysis} \label{sec:ExperimentalEvaluation_Efficiency}
\subsubsection{Time Complexity Analysis}
Table~\ref{table:TimeComplexity} shows the theoretical time complexities of all the competing algorithms. DynaMo, QCA and GreMod have different time complexities while running in different scenarios (i.e., best/worst case). As discussed in Section~\ref{sec:DynaMo_TimecomplexityAnalysis}, DynaMo has the best case time complexity when the network changes are ICEA/WI, CCEA/WI or CCED/WD, and otherwise, has the worst case time complexity. Similarly, QCA and GreMod have the best case time complexity if the network changes are ICEA or CCED, and otherwise, have the worst case time complexity. For the other algorithms, the time complexities of the best and the worst cases are identical. Below show the details about our analysis.

$\bullet$ Compared with Louvain \cite{blondel2008fast}, DynaMo has less time complexity, when the impact of the network changes of a given network snapshot on its community structure updating is small enough to ensure $m^{*}_{d} \ll m$. First, the evolutionary nature of the real-world dynamic networks assumes two consecutive network snapshots of the same network should have similar community structures. Therefore, each snapshot of a dynamic network should only result in a small part of its community structure being updated (i.e., $m^{*}_{d} \ll m$). Also, from our empirical studies, the assumption of $m^{*}_{d} \ll m$ always holds. Hence, DynaMo should be more efficient than Louvain for most of the time.

$\bullet$ Compared with Batch \cite{chong2013incremental}, DynaMo has less initialization time complexity (i.e., $O(\epsilon\cdot\frac{m}{n}) < O((\upsilon+\epsilon)\cdot\frac{m}{n})$), and different second phase time complexities (i.e., $m^{*}_{d}$ vs. $m^{*}_{b}$ ).

$\bullet$ Compared with QCA \cite{nguyen2011adaptive} and GreMod \cite{shang2014real}, who update the community structure according to certain predefined rule of each network change and one at a time (i.e., not in a batch fashion), DynaMo is more efficient if each network snapshot has more network changes, since DynaMo is capable of handling a batch of network changes.

$\bullet$ Compared with LBTR \cite{shang2016targeted}, who uses machine learning models to decide if a vertex needs to revise its community, DynaMo is more consistent and practical when dealing with different real-world networks. Since the characteristics of an dynamic network keep changing over time, LBTR has to keep updating the machine learning models to adapt the new characteristics. In such case, we have to take the training time into account. Also, the time complexity of LBTR highly depends on the machine learning algorithm used for the classification problem (e.g., $O(T_{SVM})>O(T_{LR})$).

\subsubsection{Empirical Result Studies}
Since the theoretical time complexities always depend on the ideal scenarios or extreme cases, it is necessary to conduct empirical studies using real-world networks. To ensure the comparison is as unbiased as possible, all the algorithms are implemented using Java and running on the same environment. Fig.~\ref{fig:CET} shows the cumulative elapsed time results, and below show the details about our observations.

$\bullet$ Compared with Louvain \cite{blondel2008fast}, DynaMo obtains over 2x, 2x, 4x, 3x, 4x and 3x speed up on the series of network snapshots of Cit-HepPh, Cit-HepTh, DBLP, Facebook, Flickr and YouTube, respectively.

$\bullet$ Compared with Batch \cite{chong2013incremental}, DynaMo obtains over 3x, 5x, 2x, 7x and 5x speed up on the series of network snapshots of Cit-HepPh, Cit-HepTh, DBLP, Facebook and Flickr, respectively. DynaMo spends nearly the same amount of time as Batch on YouTube network.

$\bullet$ Compared with QCA \cite{nguyen2011adaptive}, DynaMo obtains over 2x, 2x, 4x and 5x speed up on the series of network snapshots of Cit-HepTh, Facebook, Flickr and YouTube, respectively. DynaMo is as efficient as QCA on DBLP network, and spends slightly more time on Cit-HepPh network than QCA.

$\bullet$ Compared with GreMod \cite{shang2014real}, DynaMo spends more time on most of the networks, and only performs better on the Flickr and YouTube network.

$\bullet$ Compared with LBTR \cite{shang2016targeted}, DynaMo is much more efficient than LBTR-SVM, and spends slightly more time than LBTR-LR on certain networks.

\subsection{Summary of the Experimental Evaluation}
DynaMo consistently outperforms all the other 5 dynamic algorithms on 6 real-world networks and 10,000 synthetic networks in terms of the effectiveness (i.e., modularity, NMI and ARI) of detecting communities. DynaMo has almost identical performance as Louvain in terms of the effectiveness, with only 0.27\% to 1.7\% lower modularity on certain networks. DynaMo also performs comparably well in terms of the efficiency. For instance, in terms of the cumulative elapsed time results, DynaMo outperforms Louvain, Batch and LBTR-SVM, and obtains similar performance as QCA and LBTR-LR. Even though GreMod acts more efficient than DynaMo, DynaMo is much more effective than GreMod (e.g., GreMod has the worst effectiveness performance running on nearly all datasets). In conclusion, DynaMo significantly outperformed the state-of-the-art dynamic algorithms in terms of effectiveness, and demonstrated much more efficient than the state-of-the-art static algorithm, Louvain algorithm, in detecting communities of dynamic networks, while also maintaining similar efficiency as the best set of competing dynamic algorithms.

\section{Conclusion} \label{sec:Conclusion}
In this paper, we proposed DynaMo, a novel modularity-based dynamic community detection algorithm, aiming to detect communities in dynamic networks. We also present the theoretical guarantees to show why/how our operations could maximize the modularity, while avoiding redundant and repetitive computations. In the experimental evaluation, a comprehensive comparison has been made among our algorithm, Louvain algorithm and 5 other dynamic algorithms. Extensive experiments have been conducted on 6 real world networks and 10,000 synthetic networks. Our results show that DynaMo outperforms all the other 5 dynamic algorithms in terms of the effectiveness, and is 2 to 5 times (by average) faster than Louvain algorithm.


%

%
%
%
%
%

\ifCLASSOPTIONcaptionsoff
  \newpage
\fi



%
%
%

\appendices
\section{}
\label{Appendix_A}

\subsection{Proof of Proposition~\ref{prop:1}}
  \label{Appendix_A_1}
\begin{proof}
Let $Q^{(t+1)}_{1}$ denote the new modularity value if the community structure keeps unchanged (i.e., $c_{i}=c_{j}$), and $Q^{(t+1)}_{2}$ denote the new modularity value if $i$ and $j$ are split into different communities (i.e., $c^{\prime}_{i} \subseteq c_{i}$ and $c^{\prime}_{j} = c_{i} \backslash c^{\prime}_{i}$).
\begin{equation}
    \begin{split} \label{eq:1.1}
        Q^{(t+1)}_{1}=&\frac{1}{2m+2\triangle w}\Big(\alpha_{c_{i}}+2\triangle w-\frac{(\beta_{c_{i}}+2\triangle w)^{2}}{2m+2\triangle w}\\
        &+{\sum_{c \in C}^{c \neq c_{i}}}\big(\alpha_{c}-\frac{\beta_{c}^{2}}{2m+2\triangle w}\big)\Big)
    \end{split}
\end{equation}
\begin{equation}
    \begin{split} \label{eq:1.2}
        Q^{(t+1)}_{2}=&\frac{1}{2m+2\triangle w}\Big(\alpha_{c^{\prime}_{i}}-\frac{(\beta_{c^{\prime}_{i}}+\triangle w)^{2}}{2m+2\triangle w}+
        \alpha_{c^{\prime}_{j}}\\ &-\frac{(\beta_{c^{\prime}_{j}}+\triangle w)^{2}}{2m+2\triangle w}
        +{\sum_{c \in C}^{c \neq c_{i}}}\big(\alpha_{c}-\frac{\beta_{c}^{2}}{2m+2\triangle w}\big)\Big)
    \end{split}
\end{equation}
where $\beta_{c_{i}}=\beta_{c^{\prime}_{i}}+\beta_{c^{\prime}_{j}}$.

Let $Q^{(t)}_{1}$ denote the modularity value of the ``optimal'' community structure of network snapshot $G^{(t)}$, and $Q^{(t)}_{2}$ denote its modularity value while $i$ and $j$ were split into different communities as in the calculation of $Q^{(t+1)}_{2}$.
\begin{equation}
    \begin{split} \label{eq:1.3}
        Q^{(t)}_{1}=&\frac{1}{2m}\Big(\alpha_{c_{i}}-\frac{\beta_{c_{i}}^{2}}{2m}+{\sum_{c \in C}^{c \neq c_{i}}}\big(\alpha_{c}-\frac{\beta_{c}^{2}}{2m}\big)\Big)
    \end{split}
\end{equation}
\begin{equation}
    \begin{split} \label{eq:1.4}
        Q^{(t)}_{2}=&\frac{1}{2m}\Big(\alpha_{c^{\prime}_{i}}-\frac{\beta_{c^{\prime}_{i}}^{2}}{2m}+
        \alpha_{c^{\prime}_{j}}-\frac{\beta_{c^{\prime}_{j}}^{2}}{2m} \\
        &+{\sum_{c \in C}^{c \neq c_{i}}}\big(\alpha_{c}-\frac{\beta_{c}^{2}}{2m}\big)\Big)
    \end{split}
\end{equation}

Since $c_{i}=c_{j}$ is the ``optimal'' community structure of $G^{(t)}$, we have:
\begin{equation}
    \begin{split} \label{eq:1.5}
        &Q^{(t)}_{1}-Q^{(t)}_{2} \geq 0 \\
        &\Leftrightarrow \frac{1}{2m}\Big(
        \alpha_{c_{i}} - \alpha_{c^{\prime}_{i}} - \alpha_{c^{\prime}_{j}}
        - \frac{\beta_{c^{\prime}_{i}} \beta_{c^{\prime}_{j}}}{m}
        \Big) \geq 0 \\
        &\Leftrightarrow \alpha_{c_{i}} - \alpha_{c^{\prime}_{i}} - \alpha_{c^{\prime}_{j}}
        - \frac{\beta_{c^{\prime}_{i}} \beta_{c^{\prime}_{j}}}{m} \geq 0 \\
        &\Rightarrow \alpha_{c_{i}} - \alpha_{c^{\prime}_{i}} - \alpha_{c^{\prime}_{j}}
        - \frac{\beta_{c^{\prime}_{i}} \beta_{c^{\prime}_{j}}}{m+\triangle w} \geq 0
    \end{split}
\end{equation}

By comparing $Q^{(t+1)}_{1}$ and $Q^{(t+1)}_{2}$, we get the difference of the modularity gain between the ``unchanged'' and ``split'' operations as follows:
\begin{equation}
    \begin{split} \label{eq:1.7}
    &\Big(Q^{(t+1)}_{1}-Q^{(t)}_{1}\Big) - \Big(Q^{(t+1)}_{2}-Q^{(t)}_{1}\Big) = Q^{(t+1)}_{1} - Q^{(t+1)}_{2} \\
    &=\frac{1}{2m+2\triangle w}\Big(
        \alpha_{c_{i}} - \alpha_{c^{\prime}_{i}} - \alpha_{c^{\prime}_{j}}
        - \frac{\beta_{c^{\prime}_{i}} \beta_{c^{\prime}_{j}}}{m+\triangle w}
        \Big) \\
        &+
        \frac{1}{2m+2\triangle w}\Big(
        \frac{\triangle w(2m-\beta_{c^{\prime}_{i}}-\beta_{c^{\prime}_{j}}) + (\triangle w)^{2}}{m+\triangle w}
        \Big)
    \end{split}
\end{equation}

Since $\beta_{c^{\prime}_{i}}+\beta_{c^{\prime}_{j}}=\beta_{c_{i}} \leq 2m$, $\triangle w>0$ and equation (\ref{eq:1.5}), we have $Q^{(t+1)}_{1} - Q^{(t+1)}_{2} > 0$ and the conclusion follows.
\end{proof}

\subsection{Proof of Proposition~\ref{prop:2}}
  \label{Appendix_A_2}
\begin{proof}
By Proposition \ref{prop:1}, $i$ and $j$ should belong to the same community after ICEA/WI happened between $i$ and $j$, where $c_{i}=c_{j}$. Without loss of generality, we assume $i$ and $j$ belong to community $c_{p}$, even after certain bi-split operation. Therefore,
\begin{equation}
    \begin{split} \label{eq:2.1}
    &\triangle w > \frac{m\alpha_{1}-\beta_{c_{p}}\beta_{c_{q}}}{2\beta_{c_{q}}-\alpha_{1}} \Leftrightarrow \\
    & \Bigg( \frac{1}{2m+2\triangle w}\Big(\alpha_{c_{i}}+2\triangle w-\frac{(\beta_{c_{i}}+2\triangle w)^{2}}{2m+2\triangle w}\\
        &+{\sum_{c \in C}^{c \neq c_{i}}}\big(\alpha_{c}-\frac{\beta_{c}^{2}}{2m+2\triangle w}\big)\Big)\Bigg)\\
        &- \Bigg(\frac{1}{2m+2\triangle w}\Big(\alpha_{c_{p}}+ 2\triangle w-\frac{(\beta_{c_{p}}+2\triangle w)^{2}}{2m+2\triangle w}+
        \alpha_{c_{q}}\\ &-\frac{\beta^{2}_{c_{q}}}{2m+2\triangle w}
        +{\sum_{c \in C}^{c \neq c_{i}}}\big(\alpha_{c}-\frac{\beta_{c}^{2}}{2m+2\triangle w}\big)\Big)\Bigg) < 0
    \end{split}
\end{equation}
where $\beta_{c_{i}}=\beta_{c_{p}}+\beta_{c_{q}}$, and $\{c_{p}$, $c_{q}\}$ is a bi-split of $c_{i}$.
\end{proof}

\subsection{Proof of Proposition~\ref{prop:4}}
  \label{Appendix_A_3}

\begin{proof}
Let $Q_{1}^{(t+1)}$ denote the modularity value if the community structure keeps unchanged, and $Q_{2}^{(t+1)}$ denote the modularity value if $c_{i}$ and $c_{j}$ are merged into $c_{k}$ (i.e., $c_{k} = c_{i} \cup c_{j}$). Then, we have:
\begin{equation}
    \begin{split} \label{eq:4.1}
        Q^{(t+1)}_{1}&=\frac{1}{2m+2\triangle w}\Big(\alpha_{c_{i}}-\frac{(\beta_{c_{i}}+\triangle w)^{2}}{2m+2\triangle w}+
        \alpha_{c_{j}} \\
        &-\frac{(\beta_{c_{j}}+\triangle w)^{2}}{2m+2\triangle w}
        +{\sum_{c \in C}^{c \neq c_{i}, c_{j}}}\big(\alpha_{c}-\frac{\beta_{c}^{2}}{2m+2\triangle w}\big)\Big)
    \end{split}
\end{equation}
\begin{equation}
    \begin{split} \label{eq:4.2}
        Q^{(t+1)}_{2}&=\frac{1}{2m+2\triangle w}\Big(\alpha_{c_{k}}+2\triangle w-\frac{(\beta_{c_{k}}+2\triangle w)^{2}}{2m+2\triangle w} \\
        &+{\sum_{c \in C}^{c \neq c_{i}, c_{j}}}\big(\alpha_{c}-\frac{\beta_{c}^{2}}{2m+2\triangle w}\big)\Big)
    \end{split}
\end{equation}
where $\beta_{c_{k}}=\beta_{c_{i}}+\beta_{c_{j}}$. Therefore, we have the follows:
\begin{equation}
    \begin{split} \label{eq:4.3}
    &\triangle w > \frac{1}{2} \big( \alpha_{2} + \beta_{2} - 2m + \\
    &\sqrt{(2m - \alpha_{2} - \beta_{2})^{2} + 4(m\alpha_{1}+\beta_{c_{i}}\beta_{c_{j}})} \big)\\
    & \Leftrightarrow \alpha_{1}-2 \triangle w + \frac{(\beta_{c_{i}}+ \triangle w)(\beta_{c_{j}}+\triangle w)}{m+\triangle w} < 0 \\
    & \Leftrightarrow Q^{(t+1)}_{1} < Q^{(t+1)}_{2}
    \end{split}
\end{equation}
Hence, the conclusion follows.
\end{proof}


\subsection{Proof of Proposition~\ref{prop:3}}
  \label{Appendix_A_5}
\begin{proof}
Suppose the edge weight between vertices $i$ and $j$ has been decreased, where $c_{i}=c_{j}$. Let $Q^{(t+1)}_{1}$ be the modularity value if the community structure keeps unchanged, and $Q^{(t+1)}_{2}$ be the (best case) modularity value if $i$ and $j$ are split into smaller communities (i.e., $c^{\prime}_{i} \subseteq c_{i}$ and $c^{\prime}_{j} = c_{i} \backslash c^{\prime}_{i}$).
\begin{equation}
    \begin{split} \label{eq:3.1}
        Q^{(t+1)}_{1}&=\frac{1}{2m-2\triangle w}\Big(\alpha_{c_{i}}-2\triangle w-\frac{(\beta_{c_{i}}-2\triangle w)^{2}}{2m-2\triangle w} \\
        &+{\sum_{c \in C}^{c \neq c_{i}}}\big(\alpha_{c}-\frac{\beta_{c}^{2}}{2m-2\triangle w}\big)\Big)
    \end{split}
\end{equation}
\begin{equation}
    \begin{split} \label{eq:3.2}
        Q^{(t+1)}_{2}&=\frac{1}{2m-2\triangle w}\Big(\alpha_{c^{\prime}_{i}}-\frac{(\beta_{c^{\prime}_{i}}-\triangle w)^{2}}{2m-2\triangle w}+
        \alpha_{c^{\prime}_{j}} \\
        &-\frac{(\beta_{c^{\prime}_{j}}-\triangle w)^{2}}{2m-2\triangle w}
        +{\sum_{c \in C}^{c \neq c_{i}}}\big(\alpha_{c}-\frac{\beta_{c}^{2}}{2m-2\triangle w}\big)\Big)
    \end{split}
\end{equation}
where $\beta_{c_{i}}=\beta_{c^{\prime}_{i}}+\beta_{c^{\prime}_{j}}$.
\begin{equation}
    \begin{split} \label{eq:3.3}
    &Q^{(t+1)}_{1} - Q^{(t+1)}_{2} \\
    &=\frac{1}{2m-2\triangle w}\Big(
        \alpha_{c_{i}} - 2\triangle w - \alpha_{c^{\prime}_{i}} - \alpha_{c^{\prime}_{j}}\\
        &- \frac{(\beta_{c^{\prime}_{i}}-\triangle w) (\beta_{c^{\prime}_{j}}-\triangle w)}{m-\triangle w}
        \Big) \\
    &=\frac{(w_{ij}^{\prime} - \triangle w)\Big( (2m-\alpha_{c_{i}}) +(w_{ij}^{\prime}-\triangle w)\Big) - \alpha_{c^{\prime}_{i}}\alpha_{c^{\prime}_{j}}}{2(m-\triangle w)^{2}}
    \end{split}
\end{equation}
where $w_{ij}^{\prime}=\frac{\alpha_{c_{i}} - \alpha_{c^{\prime}_{i}} - \alpha_{c^{\prime}_{j}}}{2}$, $\beta_{c^{\prime}_{i}}=\alpha_{c^{\prime}_{i}}+w_{ij}^{\prime}$, $\beta_{c^{\prime}_{j}}=\alpha_{c^{\prime}_{j}}+w_{ij}^{\prime}$.

If $i$ or $j$ has only one neighbor vertex ($j$ or $i$), then,
\begin{equation}
    \begin{split} \label{eq:3.4}
    &Q^{(t+1)}_{1} - Q^{(t+1)}_{2} \\ &=\frac{(w_{ij}^{\prime} - \triangle w)\Big( (2m-\alpha_{c_{i}}) +(w_{ij}^{\prime}-\triangle w)\Big)}{2(m-\triangle w)^{2}} > 0
    \end{split}
\end{equation}
where $w_{ij}^{\prime} > \triangle w$, $2m > \alpha_{c_{i}}$. The conclusion follows.
\end{proof}

\subsection{Proof of Proposition~\ref{prop:6}}
  \label{Appendix_A_6}
\begin{proof}
Let $Q_{i}^{(t+1)}$ and $Q_{i}^{(t)}$ be the modularities of $c_{i}$ before and after the CCED/WD scenario. Then, we have:
\begin{equation}
    \begin{split} \label{eq:6.1}
    &\triangle Q = Q_{i}^{(t+1)}+Q_{j}^{(t+1)}-Q_{i}^{(t)}-Q_{j}^{(t)}=\frac{\triangle w(\alpha_{c_{i}}+\alpha_{c_{j}})}{2m(m+\triangle w)} \\
    &+\frac{1}{4}\big(\frac{\beta_{c_{i}}}{m} - \frac{\beta_{c_{i}}-\triangle w}{m-\triangle w}\big)\big(\frac{\beta_{c_{i}}}{m} + \frac{\beta_{c_{i}}-\triangle w}{m-\triangle w}\big)\\
    &+\frac{1}{4}\big(\frac{\beta_{c_{j}}}{m} - \frac{\beta_{c_{j}}-\triangle w}{m-\triangle w}\big)\big(\frac{\beta_{c_{j}}}{m} + \frac{\beta_{c_{j}}-\triangle w}{m-\triangle w}\big)
    \end{split}
\end{equation}
Let $k=\min\{ \big(\frac{\beta_{c_{i}}}{m} + \frac{\beta_{c_{i}}-\triangle w}{m-\triangle w}\big), \big(\frac{\beta_{c_{j}}}{m} + \frac{\beta_{c_{j}}-\triangle w}{m-\triangle w}\big)\}$. Thus,
\begin{equation}
    \begin{split} \label{eq:6.2}
    \triangle Q &>
    \frac{k}{4}\big(\frac{\beta_{c_{i}}+\beta_{c_{j}}}{m} - \frac{\beta_{c_{i}}+\beta_{c_{j}}-2\triangle w}{m-\triangle w}\big)\\
    &=\frac{k\triangle w (2m-\beta_{c_{i}}-\beta_{c_{j}})}{4m(m-\triangle w)}>0
    \end{split}
\end{equation}
where $2m>\beta_{c_{i}}+\beta_{c_{j}}$, $m>\triangle w$. The conclusion follows.
\end{proof}

\subsection{Proof of Proposition~\ref{prop:7}}
  \label{Appendix_A_7}
\begin{proof}
Let $Q_{1}^{(t+1)}$ denote the modularity value while merging $i$ into community $c_{j}$, $Q_{2}^{(t+1)}$ denote the modularity value while keeping $i$ as a singleton community, and $\triangle w > 0$ denote the sum of the weights of all of $i$'s associated edges. Then, we have:
\begin{equation}
    \begin{split} \label{eq:7.1}
    Q^{(t+1)}_{1}&=\frac{1}{2m+2\triangle w}\Big(\alpha_{c_{j}}+2\triangle w -\frac{(\beta_{c_{j}}+2\triangle w)^{2}}{2m+2\triangle w} \\
        &+{\sum_{c \in C}^{c \neq c_{j}}}\big(\alpha_{c}-\frac{\beta_{c}^{2}}{2m+2\triangle w}\big)\Big)
    \end{split}
\end{equation}
\begin{equation}
    \begin{split} \label{eq:7.2}
    Q^{(t+1)}_{2}&=\frac{1}{2m+2\triangle w}\Big(\alpha_{c_{j}}-\frac{(\beta_{c_{j}}+\triangle w)^{2}}{2m+2\triangle w} \\
        &-\frac{(\triangle w)^{2}}{2m+2\triangle w}+{\sum_{c \in C}^{c \neq c_{j}}}\big(\alpha_{c}-\frac{\beta_{c}^{2}}{2m+2\triangle w}\big)\Big)
    \end{split}
\end{equation}
\begin{equation}
    \begin{split} \label{eq:7.3}
    Q^{(t+1)}_{1}-Q^{(t+1)}_{2}&=\frac{\triangle w (2m-\beta_{c_{j}}) + (\triangle w)^{2}}{2(m+\triangle w)^{2}} > 0
    \end{split}
\end{equation}
where $2m \geq \beta_{c_{j}}$. Hence, the conclusion follows.
\end{proof}

\subsection{Proof of Proposition~\ref{prop:8}}
  \label{Appendix_A_8}
\begin{proof}
Let $Q_{1}^{(t+1)}$ and $Q_{2}^{(t+1)}$ denote the modularity values while merging $i$ into community $c_{p}$ and community $c_{q}$, respectively. Suppose $\triangle w_{ip} > \triangle w_{iq}$, and $\triangle w$ denoting the sum of the weights of all of $i$'s associated edges. Then, we have:
\begin{equation}
    \begin{split} \label{eq:8.1}
    &Q^{(t+1)}_{1} - Q^{(t+1)}_{2}=\frac{1}{2m+2\triangle w}
    \Big(\alpha_{c_{p}}+2\triangle w_{ip} \\
    &-\frac{(\beta_{c_{p}}+2\triangle w_{ip})^{2}}{2m+2\triangle w} + \alpha_{c_{q}}
    -\frac{(\beta_{c_{q}}+\triangle w_{iq})^{2}}{2m+2\triangle w} \Big)\\
    &- \frac{1}{2m+2\triangle w}
    \Big(\alpha_{c_{p}}-\frac{(\beta_{c_{p}}+\triangle w_{ip})^{2}}{2m+2\triangle w} + \alpha_{c_{q}} \\
    &+2\triangle w_{iq}
    -\frac{(\beta_{c_{q}}+2\triangle w_{iq})^{2}}{2m+2\triangle w} \Big)\\
    &=\frac{(4m-2\beta_{c_{p}}) \triangle w_{ip} - (4m-2\beta_{c_{q}}) \triangle w_{iq}}{(2m+2\triangle w)^{2}} \\
    &+\frac{(4\triangle w-3\triangle w_{ip}) \triangle w_{ip} - (4\triangle w-3\triangle w_{iq}) \triangle w_{iq}}{(2m+2\triangle w)^{2}}\\
    & \geq \frac{(k_{1}+k_{2})(\triangle w_{ip} -  \triangle w_{iq})}{(2m+2\triangle w)^{2}}
    \end{split}
\end{equation}
where $k_{1}=\min\{ (4m-2\beta_{c_{p}}), (4m-2\beta_{c_{q}})\}$ and $k_{2}=\min\{ (4\triangle w-3\triangle w_{ip}), (4\triangle w-3\triangle w_{iq})\}$.

Since $2m > \beta_{c_{p}}$, $2m > \beta_{c_{q}}$, $\triangle w > \triangle w_{ip}$, $\triangle w > \triangle w_{iq}$ and $\triangle w_{ip} > \triangle w_{iq}$, we have $k_{1}>0$, $k_{2}>0$, and $Q^{(t+1)}_{1} - Q^{(t+1)}_{2}>0$. Hence, the conclusion follows.
\end{proof}

\bibliographystyle{IEEEtran}
\bibliography{IEEEabrv,DynaMo}
\begin{IEEEbiography}[{\includegraphics[width=1.05in,height=2.0in,clip,keepaspectratio]{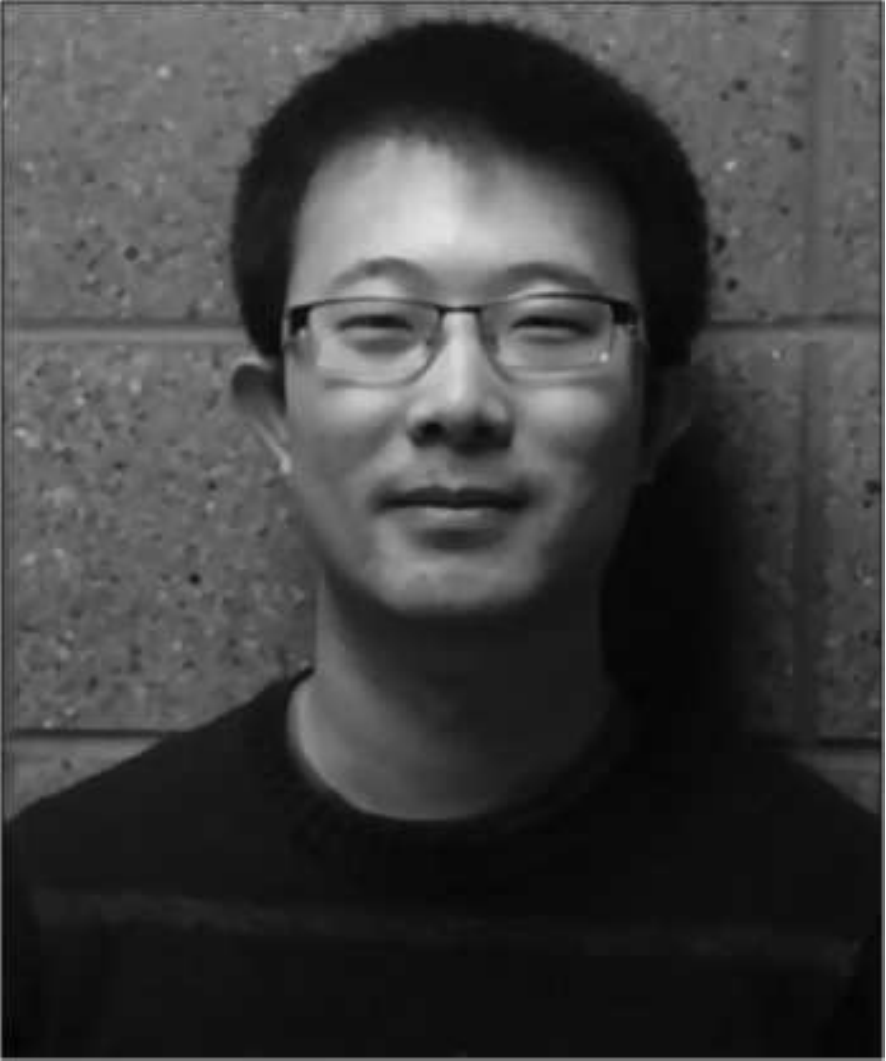}}]{Di Zhuang}
(S'15) received the B.E. degree in computer science and information security from Nankai University, China. He is currently pursuing the Ph.D. degree in electrical engineering with University of South Florida, Tampa. His research interests include cyber security, social network science, privacy enhancing technologies, machine learning and deep learning. He is a student member of IEEE.
\end{IEEEbiography}
\vskip 0pt plus -1fil
\begin{IEEEbiography}[{\includegraphics[width=1.05in,height=2.0in,clip,keepaspectratio]{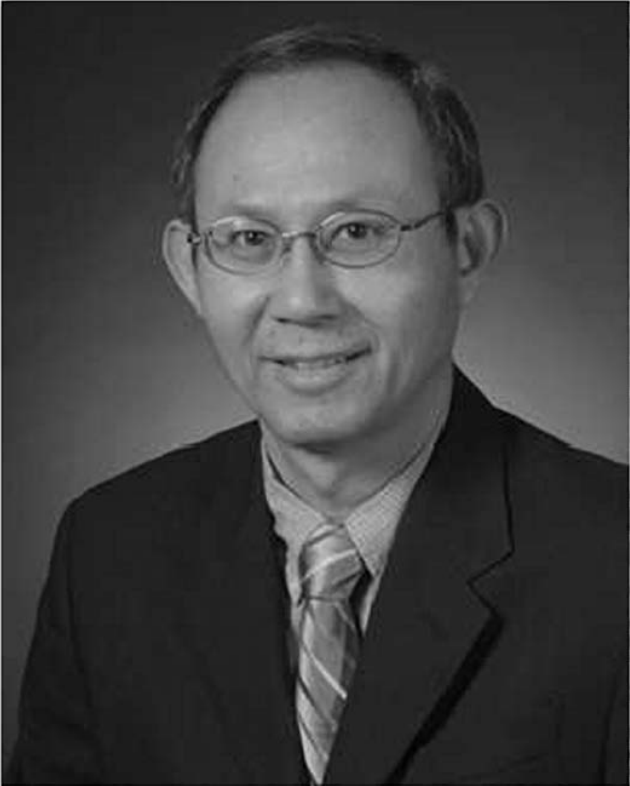}}]{J. Morris Chang}
(SM'08) is a professor in the Department of Electrical Engineering at the University of South Florida. He received the Ph.D. degree from the North Carolina State University. His past industrial experiences include positions at Texas Instruments, Microelectronic Center of North Carolina and AT\&T Bell Labs. He received the University Excellence in Teaching Award at Illinois Institute of Technology in 1999. His research interests include: cyber security, wireless networks, and energy efficient computer systems. In the last six years, his research projects on cyber security have been funded by DARPA. Currently, he is leading a DARPA project under Brandeis program focusing on privacy-preserving computation over Internet. He is a handling editor of Journal of Microprocessors and Microsystems and an editor of IEEE IT Professional. He is a senior member of IEEE.
\end{IEEEbiography}
\vskip 0pt plus -1fil
\begin{IEEEbiography}[{\includegraphics[width=1.05in,height=2.0in,clip,keepaspectratio]{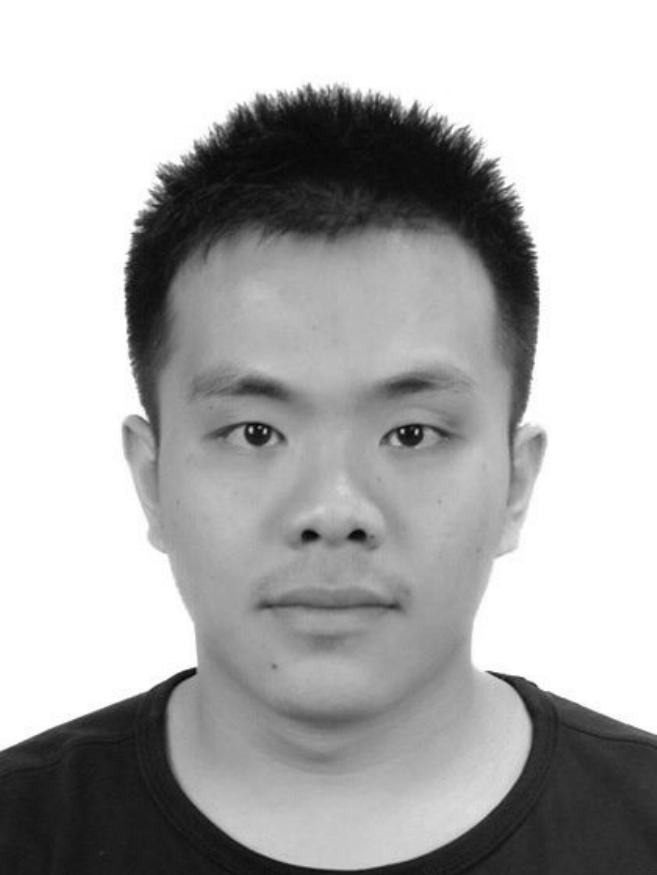}}]{Mingchen Li} received the M.S. degree in electrical engineering from Illinois Institute of Technology. He is currently pursuing the Ph.D. degree in electrical engineering with University of South Florida, Tampa. His research interests include cyber security, synthetic data generation, privacy enhancing technologies, machine learning and data analytics.
\end{IEEEbiography}
%









\end{document}